\documentclass[11pt]{article}
\usepackage{graphicx}

\usepackage[utf8]{inputenc}
\usepackage{amsmath,amsfonts,amssymb,amscd,amsthm,txfonts,hyperref}

\newtheorem{theorem}{Theorem}[section]
\newtheorem{proposition}[theorem]{Proposition}
\newtheorem{lemma}[theorem]{Lemma}

\newtheorem{remark}{Remark}[section]
\newtheorem{theo}{Theorem}

\theoremstyle{definition}
\newtheorem{definition}[theorem]{Definition}
\newtheorem{example}[theorem]{Example}

\newcommand{\R}{\varmathbb R}

\newcommand{\N}{\varmathbb N}
\newcommand{\Z}{\varmathbb Z}

\newcommand{\reff}[1]{(\ref{#1})}

\title{Wave turbulence for the BBM equation : Stability of a Gaussian statistics under the flow of BBM}
\author{Anne-Sophie de Suzzoni\footnote{D\'epartement de Math\'ematiques, Universit\'e de Cergy-Pontoise, Site de Saint Martin, 2, av Adolphe Chauvin, 95302 Cergy-Pontoise Cedex, FRANCE, e-mail : \texttt{anne-sophie.de-suzzoni@u-cergy.fr}}}

\begin{document}

\maketitle

\begin{abstract}This paper introduces a measure or statistics invariant through the flow of the Benjamin-Bona-Mahony equation and studies its stability, regarding a specific class of perturbation and in the idea of the wave turbulence theory. \end{abstract}

\tableofcontents

\section{Introduction}
 
Wave turbulence studies the evolution of some particular statistics under the flow of non linear equations with a weak non linearity. An early reference on the subject is the one by Peierls in \cite{physica} in 1929. The theory has known important developments during the sixties thanks to Zakharov, Filonenco, or Musher see \cite{zakfilo,zakde,zaktrois}, who have described the invariance of particular spectra known as the Kolmogorov-Zakharov (KZ) spectra which represents the average amplitudes to the square of waves, or the number of particles given a wavelength. In more recent works such as \cite{zakquatre}, the stability of the KZ spectrum has also been studied.

The purpose of these works is to consider each possible wavenumber of the linearised around the zero solution equation (that is to say, the non linear equation which the non linearity has been removed of) and assume what is called the random phase approximation (RPA) which presumes that the phases of the waves corresponding to these wavenumbers are initially independent from each other and taken uniformly distributed over the circle $S^1$. 

A review of general KZ spectra, that is, of the statistics such that the average of the amplitudes to the square are invariant under the flow of some PDEs can be found in \cite{first}. 

Moreover, different time scales have been observed between the studied PDE and the evolution of average quantities such as the energy, \cite{new,newde}.

Another question that arises is how does the law of the statistics itself evolve. This question appeared since the beginning of the theory in \cite{physica} and was later developed by Brout and Prigogine in \cite{probalu}. 

In more recent papers, not only the phases are supposed independent but also the modulus of the amplitudes (Random phase and amplitude assumption). One can thus wonder whether the waves remain independent as they evolve in time and as they interact due to the non linearity of the equation. A general investigation leads to the preservation of the independence and the distribution of the phases under some conditions up to corrections of order $2$ wrt a small parameter controlling the non linearity, \cite{stab}.

To be more precise, what is called a statistics is a random variable with value in $L^2$, or a space linearly spanned by the eigenmodes of the linear equation corresponding to the non linear PDE. This random variable induces a measure on $L^2$. It can also be seen as two sequences of random variables with values in $\R^+$ and $S^1$, $A_n\in \R^+$ (for the modulus of the amplitude) and $\varphi_n \in S^1$ for the phase. Then, the random initial data is given by : 

$$\sum_n A_n \varphi_n e_n$$
where $e_n$ are the eigenfunctions of the linear operator involved in the studied PDE.

Under the random phase approximation, the $\varphi_n$ are supposed initially independent from each other and from the $A_n$, and are supposed uniformly distributed over $S^1$. Under the random phase and amplitude assumption, the $A_n$ are also supposed independent from each other. The quantity introduced in \cite{stab} to study the stability of the statistics is :

\begin{equation}\label{genfunone}Z^{N}\{\lambda,\mu, t\} = \langle \prod_{i=1}^N e^{\lambda_i A_i(t)^2} \varphi_i(t)^{\mu_i} \rangle\end{equation}
where $\langle \; . \; \rangle$ denotes the mean value wrt te initial measure induced by the statistics. The integer $N$ corresponds to a certain (large) number of waves and the behaviour of $Z^N$ is studied as $N$ goes to $\infty$. As time passes by, the values of the random variable $A_n$ and $\varphi_n$ evolve and interact with each other, and thus, $Z^N$ depends also on $t$. Notice that this construction highly depends on the choice of the basis $(e_n)_n$. The stability is a control of the difference between $Z^N(t)$ and $Z^N(t=0)$. However, this generating functional has been chosen mainly for its convenience regarding the problem studied. Here, another one is taken, still for reasons of convenience, but the main purpose remains studying the law of the statistics.

Modulus invariant statistics (KZ spectra) are the ones such that $\langle A_n^2(t) \rangle = \langle A_n^2(0) \rangle$ at least ``locally" in $n$, that is for $n$ of a certain order. The solutions are of the form $\langle A_n^2 \rangle = Cn^\beta$.

In this paper, the equation onto which wave statistics are dealt with is the Benjamin - Bona - Mahony equation : 

\begin{equation} \label{bbm} \left \lbrace{ \begin{tabular}{ll}

$\partial_t \left(1-\partial_x^2 \right) u + \partial_x \left( u+ \frac{u^2}{2} \right) = 0$ & $u$ \mbox{ periodic in }$x ,t\in \R$ \\
$u|_{t=0} = u_0 \in H^s$ & \mbox{for some }$ s\geq 0$ \end{tabular}} \right. \; .\end{equation}

This equation is an alternative to KdV in the context of long wavelengths and small amplitudes water waves. The terms of second order in $u_x$ have been replaced by $-u_t$. It has been chosen because it has a so-called linear invariant, the $H^1$ norm to the square. This invariant permits to construct an initial datum belonging almost surely to $L^2$, whose law is invariant under the flow of the equation, that is, there is an invariant statistics (measure) $\mu$ on $L^2$ for the BBM equation.

For the measure $\mu$, the questions that generally arise in wave turbulence, are entirely dealt with thanks to its invariance. The squares of the amplitudes are invariant and equal to $\frac{2}{1+n^2}\sim n^{-2}$, the independence remains valid at all time, there is no time scale so to speak for the evolution of the average quantities in general since they are invariant.

This statistics will be slightly perturbed, in a way that shall be defined later, and the investigation is about the evolution of this perturbed statistics $\mu_V$, and what it implies for the evolution of the squares of the amplitudes. This measure depends on a small parameter $V$, which is a $\mathcal C^2$ function representing a potential, whose $L^\infty$ norm  and the $L^\infty$ norm of its derivatives are close to $0$.

Remark that the unknown is real here and not complex. Thus, the initial statistics considered has been chosen as the real part of a complex statistics satisfying the conditions imposed by wave turbulence, random phase amplitude or random phase approximation.

The generating functional used to study the evolution of the law of the statistics is the characteristic function, which means that the evolution of

\begin{equation}\label{genfuntwo}Z_V(\lambda, t) = E_V(e^{i\langle \lambda,\psi(t)u_0 \rangle })\end{equation}
is considered, where $E_V$ is the mean value wrt the perturbed statistics $\mu_V$, or $d\mu_V(u_0)$, $\psi(t)$ is the flow of the BBM equation, $\lambda \in L^2$, and the brackets denote the usual scalar product in $L^2$. 

In fact, there does not seem to exist quantities of type \reff{genfunone} in the context of BBM as it is a real valued context. However, \reff{genfuntwo} measures independence of the amplitudes as well as \reff{genfunone} and thus seems as natural as \reff{genfunone}.

It is known from \cite{globpos,david} that the BBM equation is globally well posed in $H^s$ for all $s\geq 0$ and there even exist bounds on the $L^2$ norm of $\psi(t)u_0$. The first thing proved here is the existence of a statistics invariant under the flow of BBM.

\begin{theo}There exists a measure $\mu$ on $L^2$ invariant under the flow of BBM. The measure $\mu$ is a Gaussian vector in infinite dimension. For all $A \subseteq L^2$ measurable (in the sense of the topological $\sigma$ algebra), 

$$\mu(\psi(t) A) = \mu(A) \; .$$\end{theo}

This statistics $\mu$ is taken such that all eigenmodes are independent from each other. The measure $\mu$ is of Gibbs type, in the spirit of the works by Lebowitz-Rose-Speer, \cite{LRS} and Bourgain, \cite{bourgain}. 

Now a small parameter $V$ is introduced, and the statistics $\mu$ is changed into a statistics $\mu_V$ which allows covariance (of order $V$) between the modes. As it happens, $\mu_V$ is built in a way that involves a slightly different linear operator $D_V= (1+V)^{-1/2}(1-\partial_x^2)^{-1}\partial_x (1+V)^{1/2}$ from the operator of BBM ($(1-\partial_x^2)^{-1}\partial_x$) and the perturbed eigenmodes (the projections onto the eigenfunctions of $D_V$) are independent from each other.

In fact, the change of statistics corresponds to a change of the equation, and the statistics $\mu_V$ is invariant under the perturbed flow, the new equation being :

$$\partial_t u_V + D_V (u_V+ \frac{(1+V)^{1/2}u_V^2}{2})=0$$
as BBM is 

$$\partial_t u +(1-\partial_x^2)^{-1}\partial_x \left(u+ \frac{u^2}{2}\right) = 0\; .$$

The flow of this equation globally exists and is noted $\psi_V$. The measure $\mu_V$ is an infinite dimensional Gaussian vector on $L^2$ with covariance operator $\sqrt{1+V} (1-\partial_x^2)^{-1} \sqrt{1+V}$.

In the end, there is an estimate regarding the characteristic functions \reff{genfuntwo}.

\begin{theo}\label{main} Let $\epsilon \in ]0, \frac{1}{2}[$, there exist two constants $C$ and $c$ such that for all $\lambda \in L^2$ and all $t\in \R$,

$$|Z_V(\lambda, t) -Z_V(\lambda,0)|\leq C||V||_\infty ||\lambda||_{L^2}|t|^{5/(2\epsilon)} e^{c|t|^{6/\epsilon -2}}$$
where $||V||_\infty$ controls the smallness of the perturbative parameter $V$ and $Z_V(\lambda, t)$ is defined by \reff{genfuntwo}.\end{theo}

\begin{remark} This result leads to the stability of the so-called KZ spectrum for this equation, that is, the mean values of the amplitudes to the square differ from their initial values only with order $||V||_\infty$ and with the same behaviour in time : for $\epsilon \in ]0,1/2[$, there exist $C,c$ such that for all $n\in \N$ $t\in \R$,

$$|E_V(|\langle \psi(t)u_0, \cos (nx)\rangle |^2) - E_V(|\langle u_0, \cos (nx)\rangle |^2)|\leq C||V||_\infty |t|^{5/(2\epsilon)} e^{c|t|^{6/\epsilon -2}}\; .$$
\end{remark}

Remark that $x\mapsto e^{ix}$ can be replaced by any $F$ as long as $F$ is smooth enough, for instance if $F$ is differentiable and its derivative is bounded.

\paragraph{Plan of the paper}
In Section 2, the existence and invariance under the BBM flow of the measure $\mu$ is proved. For that, the techniques used are the same as in \cite{moimeme}. The BBM flow is approached by finite dimensional flows, and the measure by other measures onto finite dimensional spaces, such that the conservation of the approached measures under the approached flows can be actually computed. Then, the limit is taken.

In Section 3, the meaning of the phrase "perturbation of the statistics" is given. The measure $\mu$ is a Gaussian vector of diagonal covariance matrix, the perturbed measure $\mu_V$ is also a Gaussian vector whose covariance matrix coefficients depend on the Fourier coefficients of the small $\mathcal C^2$ parameter $V$, such that it tends to the covariance of $\mu$ when $V$ goes to $0$. The measure $\mu_V$ is built in a way such that it is a priori invariant under the flow $\psi_V$ of a $V$ perturbed equation. Section 3 also introduces this equation, its invariant and a method to approach it by finite dimensional equations. 

Section 4 provides a proof of local well posedness of the perturbed equation, which is mainly a verification that the operator $D_V$ has properties in common with $(1-\partial_x^2)^{-1}\partial_x$. Then, the almost sure global well posedness is exposed along with the invariance of $\mu_V$ under the flow $\psi_V(t)$ which leads to the invariance of 

$$Z'(\lambda, t) = E_V\left( e^{i\langle \lambda , \psi_V(t)u_0\rangle }\right)\; .$$

In Section 5, the same techniques as in the first sections of \cite{globpos} are used to get bounds on the $L^2$ norms of $\psi(t)u_0$ and $\psi_V(t)u_0-\psi(t)u_0$ in order to finally prove theorem \reff{main}.

\section{Invariance of the independent Gaussian statistics}

Consider a system which can be seen as a statistical repartition of waves. In particular, look at the case when the repartition is a Gaussian onto each mode of the linear equation and those Gaussians are independent. It means that two different wavelengths are statistically independent. It evolves through the flow of the BBM equation. As it evolves, the different wavelengths interfere but the statistical repartition remains the same. Namely, the statistics is represented by a measure that is invariant through the flow.

The plan of this section comes as follow : first, the measure is defined, then, its invariance through the linear flow is proved, and then, its invariance through the BBM equation.

\subsection{Linear invariance}

The measure constructed here is a infinite dimensional Gaussian w.r.t. the Laplacian. In finite dimension, it is a Gaussian vector with a covariance matrix representing $(1-\partial_x^2)^{-1}$. This also corresponds to a Brownian motion conditionned by $2\pi$ periodicity, $u(2\pi) = u(0)$. Then, the limit is taken.

The equation \reff{bbm} admits 

\begin{equation} \label{energy} \frac{1}{2} \int u\left( 1-\partial_x^2 \right) u dx\end{equation}
as an invariant. 

Then, there exists a measure invariant through the flow of \reff{bbm}, as the action of the covariance matrix to the solution is independent from time.

\begin{definition}\label{defmea} Let $(c_n)_{n\geq 0}$ and $(s_n)_{n\geq 1}$ be the orthonormal basis of real $L^2$ with periodic conditions: 

$$c_0(x) = \frac{1}{\sqrt{2\pi }}\; ,\; c_n(x) = \frac{1}{\sqrt \pi} \cos (nx) \mbox{ and } s_n(x) = \frac{1}{\sqrt \pi} \sin (nx)\; .$$ 

Let $(g_n)_{n\geq 0} $ and $(h_n)_{n\geq 1}$ be real independent centred normalized Gaussian variables on a probability space $\Omega, \mathcal A , P$. For all $M\leq N \in \N$, call

$$\varphi_M^N : \left \lbrace{ \begin{tabular}{ll}
$\Omega \times [0,2\pi ] \rightarrow \R$ \\
$ \omega, x \mapsto \sum_{n=M}^N \left( \frac{g_n(\omega)}{\sqrt{1+n^2}} c_n(x) + \frac{1}{\sqrt{1+n^2}} s_n(x) \right) $ \end{tabular}} \right. $$
with the convention $s_0=0$ and $h_0 = 0$.

Define $\mu_M^N$ the measure onto $E_M^N $ the Hilbert subspace of $L^2$ linearly spanned by $\lbrace c_n, s_n \; | \; n=M,\hdots, N\rbrace $ that is the image of $\varphi_M^N$.
\end{definition}

In \cite{limsup}, some helpful properties for the $\varphi_M^N$ and $\mu_M^N$ are given.

\begin{proposition}\label{liminff} For any $M \geq 0$, the sequence $(\varphi_M^N)_N$ converges in $L^2(\Omega, L^2([0,2\pi ]))$.

Call $\varphi_M$ it limit, and $\mu_M$ the measure on $E_M$ the subset of $L^2$ linearly spanned by $\lbrace c_n,s_n \; |\; n\geq M \rbrace $.

As a convention, $\mu_0$ is noted $\mu$.
\end{proposition}

The following statement holds :

\begin{proposition} For any open set $U \subseteq E_M $ (for the trace topology of $L^2$), 

$$\mu_M(U) \leq \liminf_{N\rightarrow \infty } \mu_M^N (U\cap E_M^N) \; .$$  
What is more, for any $s\in [0,\frac{1}{2}[$, calling $ B_R^s$ the closed ball of centre $0$ and radius $R$ in $H^s$, which is a compact set in $L^2$ when $s>0$, it comes that : 

$$\mu((B_R^s)^c)  \leq e^{-a_s R^2}$$
where $a_s = \frac{1}{4}\left( 1+ 2\sum_{n\geq 1} \frac{1}{(1+n^2)^{1-s}}\right)$ is a constant wrt $R$.\end{proposition}

This is enough to show the invariance of the measures $\mu_M$ through the linear flow.

\begin{definition} Let $S(t)$, $t\in \R$ be the linear flow of \reff{bbm}, that is the flow of : 

\begin{equation}\label{linbbm} \partial_t \left( 1-\partial_x^2 \right) u + \partial_x u = 0 \; .\end{equation}

This flow is isometric in $L^2$, but as a matter of fact, it is also isometric in $H^s$ for all $s$ and in particular, for the $s$ that have an interest regarding the measure $\mu$, that is $s\in [0,\frac{1}{2}[$.

What is more, for all $M,N$, $S(t)E_M^N = E_M^N$ and $S(t)E_M = E_M$, and it is reversible since $S(t_1+t_2) = S(t_1)\circ S(t_2)$.
\end{definition}

Then, thanks to stability, one can acknowledge the fact that the measure on the finite dimensional subspace $E_M^N$ is invariant under the linear flow. If $u$ is written : 

$$u = \sum_{n=M}^N (a_n(t)c_n + b_n(t) s_n) $$
then

$$a_n(t) = a_n^0 \cos (\frac{-n}{1+n^2} t )-b_n^0 \sin( \frac{-n }{1+n^2} t) \mbox{ and } b_n(t) = b_n^0 \cos( \frac{-n }{1+n^2} t)  +a_n^0 \sin( \frac{-n }{1+n^2} t) \; .$$

Hence, the measure $e^{-(a^2+b^2)(1+n^2)/2} dadb$ is invariant under the change of variable $a_n^0,b_n^0 \mapsto a_n(t),b_n(t)$.

As 

$$d\mu_M^N (u) = d_M^N e^{-\sum_{n=M}^N (a_n^2 + b_n^2)(1+n^2)/2} \prod_{k=M}^N da_k db_k \; ,$$
the measure $\mu_M^N$ is invariant under the linear flow.

Making a parallel with the proof of the invariance of the measure under the linear flow in \cite{moimeme}, one can see that the proposition \reff{liminff}, the reversibility of $S(t)$ and the fact that it is isometric in $L^2$ are sufficient to prove the invariance. Hence,

\begin{proposition} For all $M$, the measure $\mu_M$ on $E_M$ is invariant under the flow of \reff{linbbm}. \end{proposition}

\begin{remark}For all $M$, $\mu$ is the measure generated by $\mu_0^{M-1}$ and $\mu_M$ on the Cartesian product $E_0^{M-1}\times E_M=L^2$.\end{remark}

\subsection{Approaching the non linear flow thanks to finite dimension}

Using now the approach by Burq-Thomann-Tzvetkov \cite{BTT} and by Burq-Tzvetkov \cite{btz}, it is possible to prove the invariance of $\mu$ under the flow of \reff{bbm}. For that, the non linear flow is approached by flows in finite dimensional spaces instead of $L^2$. The idea is that it is possible to compute functionals and measurements in finite dimension, not in $L^2$. Then, to get results on the whole space, compact convergence arguments are used.

\begin{definition} Let $\Pi_N$ be the orthogonal (on $L^2$) projector on $E_0^N$ and consider the non linear equation : 

\begin{equation} \label{finitebbm} \left \lbrace{ \begin{tabular}{ll}
$\partial_t \left( 1-\partial_x^2 \right) u + \partial_x \left( u + \Pi_N \frac{(\Pi_N u)^2}{2}\right) =0 $ \\
$u(0) = u_0 \in L^2$ \end{tabular}} \right. \; .\end{equation}

\end{definition}

Writing $u_0 = \Pi_N u_0 + (1-\Pi_N) u_0 = v_N^0 + w_N^0$, with $v_N^0 \in E_0^N$ and $w_N^0 \in E_{N+1}$, one sees that the problem \reff{finitebbm} can be reduced to a linear problem with infinite dimension on $w_N^0 $ and a non linear one with finite dimension on $v_N^0$, that is $u = v_N + w_N$, $v_N \in E_0^N$, $w_N \in E_{N+1}$ satisfying : 

$$\partial_t \left( 1-\partial_x^2 \right) v_N + \partial_x \left( v_N + \Pi_N \frac{(v_N)^2}{2}\right) = 0$$
and

$$\partial_t \left( 1-\partial_x^2 \right) w_N + \partial_x w_N  = 0\; .$$

\begin{proposition}\label{liouville}The equation 
$$\partial_t \left( 1-\partial_x^2 \right) v_N + \partial_x \left( v_N + \Pi_N \frac{(v_N)^2}{2}\right) = 0 $$
has a unique global solution on $E_0^N$ and $\mu_0^N$ is invariant under its flow, noted $\phi_N$.\end{proposition}

\begin{proof}The local uniqueness and existence of the solution is due to the fact that the non linearity is Lipschitz continuous in finite dimension. The global uniqueness and existence comes from the invariance of the $H^1$-Sobolev norm (equivalent to the $L^2$ norm in finite dimension) and then, the invariance of the Lebesgue measure from Liouville's theorem for ODEs, see \cite{liouville} for the proof and further properties of Hamiltonian flows). Indeed, write $F_N(u) = (1-\partial_x^2)^{-1}\partial_x\left( u + \Pi_N \frac{u^2}{2}\right)$. This function (on $E_0^N$) derive from a Hamiltonian, see the work of Roum\'egoux, \cite{david} for the details of the proof, of the form

$$F_N(u) = J \nabla f_N (u)$$
where $J$ is an antisymmetric operator

$$ J = (1-\partial_x^2)^{-1}\partial_x $$  
and $f_N$ is the function :

$$ f_N(u) = \int \frac{u^3}{6} \; .$$

From this Hamiltonian form, it appears that $F_N$ is divergence free. Indeed, indexing some basis of $E_0^N$ by $i$ and writing $F_N = (F_N^i)_i$ in this basis, it comes

$$\mbox{div } F_N = \sum_i \partial_i F_N^i = \sum_i (J \nabla f_N)^i = \sum_{i,j} \partial_i J^i_j (\nabla f_N)^j$$

$$ = \sum_{i,j} J^i_j \partial_i \partial_j f_N $$
and $J$ being antisymmetric, this sum is zero, $F_N$ is divergence free.

Now, since $F_N$ is divergence free, the Jacobian of $\phi_N(t)$ does not depend on $t$, and so it is $1$, the Lebesgue measure is invariant under the flow, which is the Liouville theorem. Indeed, its proof gives

$$D_t \left( \mbox{jac } \phi_N(t) (u_0)\right) = D_t\left( \mbox{det } (d_{u_0} \phi_N(t) )\right) = (D_{d_{u_0}\phi_N (t)} \mbox{det})\circ \left( D_t (d_{u_0} \phi_N(t) )\right)$$
and $D_t$ and $d_{u_0}$ commute so

$$ D_t (d_{u_0} \phi_N(t)  )= d_{u_0} (D_t \phi_N(t) ) = d_{u_0} F_N \circ \phi_N(t) = d_{\phi_N(t)u_0} F \circ d_{u_0} \phi_N(t)$$

$$D_t \left( \mbox{jac } \phi_N(t) (u_0)\right) = \mbox{ Tr}\left( (d_{u_0}\phi_N(t))^{-1} \circ d_{\phi_N(t)u_0 F}\circ d_{u_0}\phi_N(t)\right)$$

$$= \mbox{ Tr}\left( d_{\phi_N(t)u_0 F} \right) = \mbox{ div } F (\phi_N(t)u_0) = 0 \; .$$

Then, as the $H^1$ norm is invariant under the flow, 

$$d\mu_0^N(u) = d_0^N e^{-\frac{1}{2}\int u(1-\partial_x^2) u}dL(u)$$
is also invariant under the flow. \end{proof}

\begin{proposition}The measure $\mu = \mu_0^N \otimes \mu_N$ is invariant under the flow of \reff{finitebbm}, noted $\psi_N$. \end{proposition}

\begin{proof} Let $A \subseteq E_0^N$ and $B \subseteq E_{N+1}$ $\mu_0^N$ and $\mu_N$ measurable respectively. Then,

$$\mu( \psi_N (t) (A\times B) ) = \mu((\phi_N(t) A) \times (S(t) B)) $$

$$ = \mu_0^N (\phi_N(t) A) \mu_{N+1} (S(t) B) = \mu_0^N(A)\mu_{N+1}(B)$$

$$= \mu( A\times B) $$
thanks to the invariance of $\mu_0^N$ under $\phi_N$ and of $\mu_{N+1}$ under $S(t)$. 

As the proposition holds for every Cartesian products, it holds on all measurable sets.\end{proof}

\subsection{Invariance under the non linear flow}

\begin{definition}For all $T \in \R_+$, set

$$X_T^s = \mathcal C ([-T,T],H^s)$$

normed by $||\; . \; ||_{L^\infty_t,H^s_x}$.\end{definition}

The following lemma comes from \cite{david} : 

\begin{lemma}\label{locex} Let $s\geq 0$. There exists a constant $C_s$ depending only on (and increasing with) $s$ such that the flow $\psi$ of the BBM equation \reff{bbm} is defined on $[-T,T]\times B_R^s$ , as long as $T < \frac{1}{C_s R}$. Moreover, if $(t,u_0) \in [-T,T]\times B_R^s$, then

$$||\psi(t)(u_0)||_{H^s}, ||\psi_N(t) (u_0)||_{H^s} \leq 2R$$
and, calling, for $||u||_{X^s_T}\leq 2R$,

$$A(u)(t) = S(t)u_0 - \frac{1}{2} \int_{0}^t S(t-s) (1-\partial_x^2)^{-1} \partial_x u^2 ds \; ,$$
for all $u,v$,

$$||A(u)-A(v)||_{X^s_t} \leq 2C_s R t ||u-v||_{X_T^s} \; ;$$
\end{lemma}

The sequence $\psi_N(t)(u_0)$ converges uniformly in $u_0 \in B_R^s$, $s>0$ for the topology of $X^0_T$ with a suitable $T$.

\begin{lemma}\label{loconv}Let $s\in ]0, \frac{1}{2}[$ and $R > 0$. Let $\epsilon > 0$, there exists $N_0 \in \N$ such that for all $N \geq N_0$, all $u_0 \in B_R^s$ and all $t \in [- \frac{1}{3 C_s R}, \frac{1}{3C_s R}]$,

$$||\psi(t) u_0 - \psi_N (t) u_0||_{L^2} \leq \epsilon \; .$$

\end{lemma}

\begin{proof} Let $u_0 \in B_R^s$. Call $u= \psi(t) u_0$ and $u_N = \psi_N(t) u_0$. Then, $u$ is a fix point for $A$ and $u_N$ for $A_N$ such that : 

$$ A_N(v)(t) = S(t) u_0 - \frac{1}{2} \int_{0}^t S(t-s) (1-\partial_x^2)^{-1} \partial_x \Pi_N (\Pi_N v(s))^2 ds$$
that is

$$A_N(v)(t)-S(t)u_0 = \Pi_N \left( A(\Pi_N v)(t) - S(t) u_0\right) \; .$$

Thus,

$$u-u_N = A(u) -A_N(u_N) = A(u)-S(t)u_0 - (A_N(t)u_N - S(t)u_0) = A(u)-A(\Pi_N u_N) + (1-\Pi_N)(S(t)u_0) \; .$$

Hence, with $T=\frac{1}{3C_s R}$,

$$||u-u_N||_{X^0_T} \leq ||(1-\Pi_N)S(t)u_0||_{L^2}  + ||A(u)-A(\Pi_N u)||_{X^0_T} \leq N^{-s}||S(t) u_0||_{H^s}+2C_0 T R ||u-\Pi_N u_N||_{X^0_T}\; .$$

Then,

$$||S(t) u_0||_{H^s} = ||u_0||_{H^s} \leq R$$
and 

$$||u-\Pi_N u_N||_{X^0_T} \leq ||u-u_N||_{X^0_T} + ||u_N - \Pi_N u_N||_{X^0_T} \leq  ||u-u_N||_{X^0_T} + N^{-s}||u_N||_{X^s_T} \; .$$

Finally,

$$||u-u_N||_{X^0_T} \leq \frac{2}{3} ||u-u_N||_{X^0_T} + N^{-s}\left( \frac{2}{3} +R\right)$$
so there exists $N_0$ depending only on $s$ and $R$ such that for all $N\geq N_0$ : 

$$||u-u_N||_{X^0_T} \leq \epsilon \; .$$ \end{proof}

\begin{lemma}\label{locinv} Let $s\in ]0,\frac{1}{2}[$ and $R> 0$. Let $A$ be a measurable set of $L^2$ included in $B_R^s$. Let $T= \frac{1}{3C_s R}$ , for all $t\in [-T,T]$,

$$\mu(\psi(t)A)= \mu(A) \; .$$\end{lemma}

\begin{proof} Suppose that $A$ is closed. There exists $N_0$ such that for all $N\geq N_0$ and all $u_0 \in A$,

$$||\psi(t)u_0 -\psi_N(t)u_0||_{L^2} \leq \epsilon\; .$$
for all $N\geq N_0$, $\psi(t) A \subseteq \psi_N(t) A + B^0_\epsilon$. Then, $\psi_N$ satisfies, thanks to the continuity and the reversibility of the local flow.

$$\psi_N(t) A + B^0_\epsilon = \psi_N(t) \psi_N(-t) (\psi_N(t) A + B_\epsilon^0) \subseteq \psi_N(t) (\psi_N(-t)\psi_N(t) A +B^0_{C\epsilon})$$

$$\psi_N(t) A + B^0_\epsilon\subseteq \psi_N(t) (A + B^0_{C\epsilon})$$
with a constant $C$ independent from $N$. So,

$$\mu(\psi(t) A) \leq \mu\left(\psi_N(t) (A+B^0_{C\epsilon})\right) = \mu(A + B^0_{C\epsilon})$$
and, with $\epsilon \rightarrow 0$, thanks to the dominated convergence theorem

$$\mu(\psi(t) A) \leq \mu(A)\; .$$

For the reverse inequality, $\psi_N(t) A \subseteq \psi(t) A + B_\epsilon^0$, so

$$\mu(\psi(t) A + B_\epsilon^0) \geq \mu ( \psi_N(t) A) = \mu(A) \;.$$

If $A$ is open, then $A^c$ the complementary of $A$ in $B_R^s$ is closed and included in $B_R^s$. So, the local invariance is true for open sets. Then, the flow being unique and reversible, it is also true for countable disjoint unions, and so for all measurable sets.\end{proof}

Build now a set set onto which $\mu$ is invariant under the BBM flow and prove that it is of full measure. Then, as it is of full measure, $\mu$ is invariant under the flow of BBM.

\begin{definition} Let $s\in]0,\frac{1}{2}[$ and $R> 0$. Let $R_k = \sqrt{k+1} R$ and  $t_k = \frac{1}{3C_s \sqrt{k+1} R}$ and $T_0=0$, $T_{n+1} = \sum_{k= 0}^{n}t_k$. Call $A_N^n(R) = \psi_N(T_n)^{-1}(B_{R_n}^s)\cup \psi_N(-T_n)^{-1}(B_{R_n}^s)$, 

$$A_N(R)= \bigcap_{n\geq 0} A_N^n(R)$$

$$A(R) = \limsup_{N\rightarrow \infty } A_N(R) \; .$$
\end{definition}

\begin{proposition} There exists two constants $C>0$ and $a>0$ such that for all $R>2$ , $\mu(A(R)^c) \leq Ce^{-aR^2}$ . \end{proposition}

\begin{proof}Indeed,

$$A(R)^c = \liminf A_N(R)^c$$

$$\mu(A(R)^c)\leq \liminf \mu(A_N(R)^c)$$
then, 

$$\mu(A_N(R)^c) \leq \sum_{n\geq 0} \mu(A_N^n(R)^c)$$
and

$$\mu(A_N^n(R)^c) = \mu(\psi_N(T_n)^{-1} (B_{R_n}^s)^c) +  \mu(\psi_N(-T_n)^{-1} (B_{R_n}^s)^c) \leq 2 \mu((B_{R_n}^s)^c)\leq 2e^{-a(n+1) R^2}$$
thanks to the invariance of $\mu$ under $\psi_N$.

$$\mu(A_N(R)^c) \leq 2\sum_{n\geq 0} (e^{-aR^2})^{(n+1)} \leq Ce^{-aR^2}  $$
with $C$ independent from $R$. Hence,

$$\mu(A(R)^c) \leq Ce^{-aR^2}$$.\end{proof}

\begin{theorem}\label{blablav}Let $C$ be a $\mu$ measurable set of $L^2$. Then, for all $t \in \R$,

$$\mu(\psi(t)C) = \mu(C) \; .$$\end{theorem}

\begin{proof} Let $C_R = C \cap A(R)$. As $A(R)$ is $\mu $ measurable, $C_R$ is also measurable and included in $A(R)$. 

By induction over $n$, $\psi(\pm T_n) C_R \subseteq B_{R_n}^s$ and for all $t \in [-T_n,T_n]$,

$$\mu( \psi(t) C_R) = \mu(C_R) \; .$$

Indeed, for $n=0$, $T_n= 0$, so $\mu(\psi(0) C_R)  = \mu(C_R)$. Then, as $C_R \subseteq A(R)$, for all $u\in C_R$, there exits a sequence $N_k \rightarrow \infty$ such that $u \in A_{N_k}(R)$, that is $\psi_{N_k}(T_n)(u) \in B_{R_n}^s$ for all $n$. In particular, for $n =0$, 

$$u = \psi_{N_k}(0)(u) \in B_{R_0}^s $$
and it will also appear by induction that for all $n$, $\psi_{N_k}(T_n) u$ converges in $L^2$ toward $\psi(T_n)(u)$ when $k$ goes to $\infty$.

For $n\rightarrow n+1$, suppose that $\psi(T_n) C_R \subseteq B_{R_n}^s$ and for all $t \in [-T_n,T_n]$, $\mu( \psi(t) C_R) = \mu(C_R)$ . As $\psi(T_n) C_R \supseteq B_{R_n}^s$ and $\psi(-T_n) C_R \subseteq B_{R_n}^s$, thanks to lemma \reff{locinv}, for all $t\in [0, t_n]$,

$$\mu(\psi(t) \psi(T_n) C_R) = \mu(\psi(T_n) C_R)= \mu(C_R)$$
and

$$\mu(\psi(-t) \psi(-T_n) C_R) = \mu(\psi(-T_n) C_R)= \mu(C_R)\; ,$$
as $T_{n+1} = T_n + t_n$, the invariance is true for $t\in [-T_{n+1},T_{n+1}]$. Then for all $u$ in $C_R$,

$$||\psi(T_{n+1})(u)-\psi_{N_k}(T_{n+1})(u)||_{L^2}=||\psi (t_n) \psi(T_n)(u) - \psi_{N_k}(t_n)\psi(T_n)(u)||_{L^2}+ $$

\hspace{2cm} $$||\psi_{N_k}(t_n)(\psi (T_n)(u)) -\psi_{N_k}(t_n)(\psi_{N_k}(T_n)(u)||_{L^2} \; .$$

Thanks to lemma \reff{locex}, there exists a constant independent from $N$ such that $||\psi_{N_k}(t_n)(u) -\psi_{N_k}(t_n)(v)||_{L^2}\leq C ||u-v||_{L^2}$ as long as $u,v \in B^0_{R_n} \subset B^s_{R_n}$, so

$$||\psi_{N_k}(t_n)(\psi (T_n)(u)) -\psi_{N_k}(t_n)(\psi_{N_k}(T_n)(u)||_{L^2}\leq C ||\psi (T_n)(u)-\psi_{N_k} (T_n)(u)||_{L^2} \rightarrow 0$$
by induction hypothesis and

$$||\psi (t_n) \psi(T_n)(u) - \psi_{N_k}(t_n)\psi(T_n)(u)||_{L^2} \rightarrow 0$$
when $k,N_k \rightarrow \infty$ thanks to lemma \reff{loconv}. So, $\psi_{N_k}(T_{n+1})(u) \in B_{R_{n+1}}^s $ converges toward $\psi(T_{n+1})u$ in $L^2$ and as $B_{R_{n+1}}^s$ is compact in $L^2$, $\psi(T_{n+1})u \in B_{R_{n+1}}$.

The induction is proved.

Then, for all $t \in \R$, as $T_{n} = \sum_{k=1}^{n}\frac{1}{3C_s R\sqrt k} \rightarrow \infty$, there exists $n$ such that $t\in [-T_n,T_n]$, so

$$\mu(\psi(t)C_R) = \mu(C_R) \; .$$

Finally,

$$\mu(\psi(t) C) \geq \mu(\psi(t)C_R) = \mu(C_R)$$
and

$$\mu(C) \leq \mu(C_R) + \mu(A(R)^c) \leq \mu(C_R)+Ce^{-aR^2}\leq \mu(\psi(t)C) +Ce^{-cR^2}$$
it comes

$$\mu(\psi(t)C)\geq \mu(C)$$
and

$$\mu(C) = \mu(\psi(-t)\psi(t) C) \geq \mu(\psi(t) C)\; ;$$
hence

$$\mu(\psi(t)C) = \mu(C) \; .$$ \end{proof}

\section{A new measure and new equations}

Now that the statistics $\mu$ has been proved to remain invariant through the BBM flow, perturb it a little bit and see if $\mu$ is stable. That is, build a statistics $\mu_V$ depending on a small parameter $V$ close to $0$ and analyse the evolution in time of $\mu_V$, see that it remains close to its initial statistics, and so close to $\mu$.

The new statistics $\mu_V$ will admit two different interpretations, depending on the point of view. First, regarding the measure itself, it adds some correlations between the wavelengths. With the statistics $\mu$, the different wavelengths were all independent from each other (the Gaussians had been taken independent), with $\mu_V$ two different wavelengths will be all the more depending on each other that their wavelengths are close.

The statistics $\mu_V$ are also the invariant statistics for the BBM equation onto which the unknown $u$ has been replaced by $ \sqrt{1+V}u$. Developing this expression to the first order in $V$, one gets a new equation corresponding to adding some external potential or a dispersive term, like frictional or shearing resistance.

\subsection{Perturbation of the measure}

The measure $\mu$ shall now be perturbed.

\begin{definition}Let $V$ be a $\mathcal C^2$ periodic function. Set $||\; .\; ||_\infty = ||\; .\; ||_{L^\infty}+|| \partial_x\; . \; ||_{L^\infty}+ ||\partial_x^2 \; .\; ||_{L^\infty}$ and suppose that $||V||_\infty \leq 1/2$. The operator multiplication by $V$, also noted $V$, is defined from $L^2$ to $L^2$ and its norm satisfies : 

$$||V||_0 = ||V||_{\mathcal L(L^2,L^2)} \leq ||V||_{L^\infty} \; .$$

\end{definition}

\begin{proposition} As $||V||_\infty$ is strictly less than $1$ and self adjoint, the operator on $L^2$

$$\left( 1 + V \right) ^{-1/2}$$
is well defined and its norm is less than 

$$||(1+V)^{-1/2}||_0 \leq \frac{1}{\sqrt{1-||V||_{L^\infty}}}\leq \sqrt 2 \; .$$ 
\end{proposition}

\begin{remark}The function $V$ is the small parameter by definition, but if one looks at $g=\sqrt{1+V}-1$, it is also a small parameter in the same norm.\end{remark}

\begin{definition} Let $B_N$ be the matrix of $\Pi_N(1 +V)^{-1/2} H^{-1}\Pi_N$ in the orthonormal basis 

$$(c_0,c_n,s_n)_{1\leq n \leq N}$$ 
where $H  = (1-\partial_x^2)^{1/2}$, that is the coefficients of $B_N$ are given by

$$ (B_N)_{n,m} = \left \lbrace{ \begin{tabular}{llll}
$\langle c_n , (1 +V)^{-1/2} H^{-1} c_m\rangle $ & \mbox{ if }$n,m \leq N$\\
$\langle c_n , (1 +V)^{-1/2} H^{-1} s_{m-N}\rangle $ & \mbox{ if }$n\leq N\; ,\; m\geq N+1$\\
$\langle s_{n-N} , (1 +V)^{-1/2} H^{-1} c_m\rangle $ & \mbox{ if }$m \leq N\; , \; n\geq N$\\
$\langle s_{n-N} , (1 +V)^{-1/2} H^{-1} s_{m-N}\rangle $ & \mbox{ otherwise.}\end{tabular}} \right. $$
where $\langle \; ,\; \rangle$ is the scalar product in $E_N$.
\end{definition}

\begin{definition}Call $g^N = (g_0,\hdots ,g_N,h_1,\hdots , h_N)$ and let $\alpha^N = B_N g^N = (\alpha_0^N,\hdots ,\alpha_N^N , \beta_0^N,\hdots , \beta_N^N )$. The vector $\alpha^N$ is a real centred Gaussian vector of covariance matrix : 

$$ B_N B_N^* \; .$$ \end{definition}

\begin{definition}Set

$$\varphi_V^N = \sum_{n=0}^N \alpha_n^N c_n + \sum_{n=1}^N \beta_n^N s_n \; .$$

The map $\varphi_V^N$ defines a measure $\mu_V^N$ on $E_0^N$.\end{definition}

\begin{proposition}Let $0\leq s < \frac{1}{2}$. The sequence $\varphi_V^N$ converges in $L^2(\Omega , H^s)$. Its limit is called $\varphi_V$ and induces a measure on $L^2$ called $\mu_V$. Besides, calling $E_V$ the mean value wrt $\mu_V$, $E_V(||u||_{L^2}^2)$ is uniformly (in $V$) bounded. \end{proposition}

\begin{proof} Let $N\geq M > 0$. The norm of $\varphi_V^N - \varphi_V^M$ is such that : 

$$||\varphi_V^N - \varphi_V^M ||_{L^2_\omega , H^s_x}^2  = E(||\varphi_V^N - \varphi_V^M ||_{H^s}^2)$$

$$ = E\left(\sum_{n=M+1}^N (1+n^2)^s(\alpha_n^N)^2+ (1+n^2)^s(\beta_n^N)^2 + \sum_{n=0}^M (1+n^2)^s(\alpha_n^N - \alpha_n^M)^2 +(1+n^2)^s(\beta_n^N - \beta_n^M)^2 \right)\; .$$

The first sum corresponds to the trace:

$$Tr ((1-\Pi_M )H^{s}B_N^* B_NH^s(1-\Pi_M )) = Tr \lbrace (1-\Pi_M) H^{s-1} \Pi_N (1+V)^{-1/2}\Pi_N (1+V)^{-1/2} H^{s-1}(1-\Pi_M)\rbrace $$

$$ = Tr \left( (1-\Pi_M)H^{2(s-1)} \Pi_N (1+V)^{-1/2}\Pi_N (1+V)^{-1/2}\right) $$

$$\leq Tr\left( (1-\Pi_M)H^{2(s-1)}\right) ||\Pi_N (1+V)^{-1/2}\Pi_N (1+V)^{-1/2}||$$

$$\leq \sum_{n=M+1}^N \frac{1}{(1+n^2)^{1-s}} \frac{2}{1-||V||_{L^\infty}}$$

$$\leq 4 \sum_{n\geq M+1} \frac{1}{(1+n^2)^{1-s}} \; ,$$
which goes to $0$ when $M\rightarrow \infty$.

The second is : 

$$Tr(\Pi_M H^s(B_N-B_M)^*(B_N-B_M)H^s\Pi_M ) = $$

\hspace{2cm} $$Tr ( \Pi_M H^{s-1}( (1+V)^{-1/2}(\Pi_N-\Pi_M) (1+V)^{-1/2}\Pi_M) H^{s-1}\Pi_M)$$
since $H^{-1},\Pi_N$ and $\Pi_M$ commute.

Then, use the fact that $\Pi_M = (\Pi_M - \Pi_{M/2} )+ \Pi_{M/2}$ to get that the trace is less than :

$$||(1+V)^{-1/2}||_{L^2}^2 Tr ((\Pi_N-\Pi_{M/2})H^{2(s-1)}) + Tr(\Pi_{M/2} ( (1+V)^{-1/2}(\Pi_N-\Pi_M) (1+V)^{-1/2}\Pi_{M/2}))\; .$$

The first trace is less than

$$C\sum_{m> M/2} \frac{1}{(1+m^2)^{1-s}}$$
which goes to $0$ when $M\rightarrow \infty$, the second is less than the sum to the square of the Fourier coefficients of $(1+V)^{-1/2}$ of wavelengths bigger than $M/2$. Indeed, if $g\in E_0^{M/2}$ and $h \in E_{M}$ then $hg \in E_{M/2+1}$. As $(1+V)^{-1/2}$ is $\mathcal C^1$, the series of its Fourier coefficients is absolutely convergent and thus the sum of its Fourier coefficients of wavelengths bigger than $M/2$ goes to $0$ when $M$ goes to $\infty$.

So, the sequence $\varphi_V^N$ is a Cauchy sequence in $L^2(\Omega , H^s)$, hence it converges toward a certain $\varphi_V$ in $H^s$.

What is more,

$$E_V( ||u||_{L^2}^2) = E( ||\varphi_V||_{L^2}^2) \leq 4 \sum \frac{1}{1+n^2}\; .$$ \end{proof}

\begin{example}The covariance between two waves is given by $E(\alpha^N_n \alpha^N_m)$ or the mean value of any combination of $\alpha^N_k$ and $\beta^N_l$ with $k,l = n$ or $m$. In particular,

$$E(\alpha^N_n \alpha^N_m) = (B_N B_N^*)_{n,m} = \sum_{k=0}^{2N+1} (B_N)_{n,k} (B_N)_{m,k} $$
$$= \sum_{k=0}^N \frac{1}{1+k^2} \langle c_nc_k,(1+V)^{-1/2} \rangle \langle c_mc_k,(1+V)^{-1/2} \rangle + \sum_{k=1}^N \frac{1}{1+k^2} \langle c_ns_k,(1+V)^{-1/2} \rangle \langle s_mc_k,(1+V)^{-1/2} \rangle$$
which involves $k$ bigger than $|n-m|/2$ or Fourier coefficients of $(1+V)^{-1/2}$ of wavelengths bigger than $|n-m|/2$. So, the dependence between two waves decreases quite quickly when the difference between the wavelengths increases.

Remark that as $(1+V)^{-1} = \sum (-V)^k$ since $V$ is small, then, for $n\neq m$, $E(\alpha_n^N \alpha_m^N)$ has no zero order in $V$, it is at least as small as $V$ itself. \end{example}

\subsection{Perturbation of the flow}

Now, there is an equation whose flow is invariant under the perturbed measure. In finite dimensional approximation, the linear operator $(1-\partial_x^2)$ is replaced by $(B_N^{-1})^*B_N^{-1}$ (the matrix that appears in the law of $\alpha^N$ as $\Pi_N (1-\partial_x^2)\Pi_N$ was the matrix that appeared in the law of $g^N$) on $E_0^N$ and $(1-\Pi_N )(1-\partial_x^2)(1-\Pi_N)$ on its orthogonal.

\begin{definition} Let $W_N$ be the operator on $E_0^N$ whose matrix is $(B_N^{-1})^* B_N^{-1}$ in the basis

$$\lbrace c_0,\hdots ,c_N, s_1,\hdots ,s_N \rbrace $$
and $V_N$ the operator $\Pi_N \sqrt{1+V} \Pi_N$ on $E_0^N$ such that $W_N = V_N (1-\partial_x^2) V_N$.

On $E_0^N$, the law of $\mu_V^N$ is given by : 

$$d\mu_V^N \left( u := \sum a_n c_n + b_n s_n\right) = d_N^V e^{-\frac{1}{2} \int u W_N u } \prod da_n db_n $$
where 

$$d_N^V = \sqrt{\mbox{det } W_N} (2\pi)^{-(2N+1)/2}$$
is a normalization factor.
\end{definition}

\begin{remark}The operator $V_N$ has an inverse on $E_0^N$ and satisfies for all $s,N, V, u\in E_0^N$ : 

$$||V_N u||_{H^s} \leq 2 ||u||_{H^s} \mbox{ and } ||V_N^{-1}u||_{H^s} \leq 4 ||u||_{H^s} \; .$$
\end{remark}

\begin{proof}Call $I_N$ the identity of $E_0^N$ and remark that 

$$||(V_N-I_N)u||_{H^s} \leq ||\Pi_N \sqrt{1+V} -1||_{L^\infty} ||u||_{H^s}$$
and that

$$||\Pi_N \sqrt{1+V} -1||_{L^\infty} \leq \frac{||V||_\infty}{\sqrt{ 1 - ||V||_\infty}} \leq \sqrt 2 ||V||_\infty$$
so

$$||V_Nu||_{H^s} \leq \left( 1+ \frac{\sqrt 2}{2} \right) ||u||_{H^s} \leq 2 ||u||_{H^s}$$
and

$$||V_N^{-1}u||_{H^s} \leq \frac{1}{1 - \sqrt 2/ 2} ||u||_{H^s} \leq 4 ||u||_{H^s}\; .$$\end{proof}

\begin{proposition} The equation

\begin{equation}\label{perturbeq} \left \lbrace{ \begin{tabular}{ll}
$\partial_t W_N u + V_N \partial_x (V_N u+ \Pi_N \frac{(V_Nu)^2}{2} ) = 0$ \\
$u|_{t=0} = u_0 \in E_0^N$ \end{tabular}} \right. \end{equation}
admits a unique global solution $u(t) = \phi_V^N(t)(u_0) $ and 

$$E_V (t) = \frac{1}{2} \int u(t) W_N u(t) $$
is invariant under this flow, it does not depend on time $t$.
\end{proposition} 

\begin{proof}
The existence and uniqueness of the local flow is given by Cauchy-Lipschitz theorem, as 

$$W_N^{-1} = V_N^{-1} H^{-2} V_N^{-1} $$
and the derivation of $E_V$ gives

$$\dot E_V  = \int u \partial_t (W_N u(t)) = - \int u V_N \partial_x (V_N u + \Pi_N \frac{(V_Nu)^2}{2}) $$

$$ \int \partial_x (V_N u) \left( V_N u + \frac{(V_N u)^2}{2}\right) = \int \partial_x \left( \frac{(V_Nu)^2}{2} + \frac{(V_N u)^3}{6} \right) = 0$$
since $W_N$ and $V_N$ are self-adjoint and $\partial_t $ and $W_N$ commute.

What is more, as $||V||_\infty \leq 1/2$, $W_N$ is strictly positive, then $\sqrt{E_V}$ is a norm on $E_0^N$ equivalent to all other norms on $E_0^N$, so the flow is global. \end{proof}

\begin{proposition}\label{invarn} The measure $\mu_V^N$ is invariant under $\phi_V^N$. \end{proposition}

To prove this proposition, Liouville's theorem is used and so it is required to give the equation \reff{perturbeq} its Hamiltonian form.

\begin{lemma} The equation \reff{perturbeq} admits a Hamiltonian formulation. \end{lemma}

\begin{proof} Call $J_N$ the operator on $E_0^N$ : 

$$J_N = W_N^{-1} V_N \partial_x V_N^{-1} =V_N^{-1} H^{-2}\partial_x V_N^{-1} \; .$$

This operator is antisymmetric, since $V_N$ and $H$ are self-adjoint, $\partial_x$ is antisymmetric and $H$ and $\partial_x$ commute.

The equation \reff{perturbeq} can be written : 

$$\partial_t u + J_N \left( V_N^2 u + V_N \Pi_N \frac{(V_N u)^2}{2} \right) = 0\; .$$

Writing $u = \frac{1}{\sqrt{2\pi}}\sum_{n=-N}^N u_n e^{inx}$ and $(V_N)^k_j = \frac{1}{2\pi}\int e^{-ikx}V_N (e^{ijx})= \frac{1}{2\pi} \int e^{-i(k-j)x} \sqrt{1+V}$, call 

$$H(u_{-N},\hdots, u_N) = H_1 + H_2$$
with

$$H_1 = \frac{1}{2} \int (V_N u)^2 \mbox{ and } H_2 = \int \frac{1}{6} \int (V_N u)^3 \; ,$$
and 

$$F^n (u_{-N},\hdots , u_N) = F^n_1 +F^n_2$$
with

$$F^n_1 = \frac{1}{\sqrt{2\pi}}\int e^{-inx} V_N^2 u \mbox{ and } F^n_2 = \frac{1}{\sqrt{2\pi}} \int e^{-inx} V_N \Pi_N \frac{(V_N u)^2}{2}\; .$$

The function $H_1$ can be rewritten : 

$$H_1 = \frac{1}{2}\sum_{k+l=0}  \sum_{n,m=-N}^N 1_{|k|\leq N}1_{|l|\leq N}(V_N)^k_n (V_N)^l_m u_n u_m$$
so its complex derivative wrt $u_n$ is : 

$$\frac{dH_1}{du_n} = \sum_{k+l = 0} \sum_{m=-N}^N 1_{|k|\leq N}1_{|l|\leq N}(V_N)^k_n (V_N)^l_m u_m\; .$$

As $V_N$ is self-adjoint : $(V_N)^k_n = \overline{ (V_N)^n_k}$. As $V_N$ and $u$ are real, $(V_N)^l_m = \overline{(V_N)^{-l}_{-m}}$ and $u_m = \overline{u_{-m}}$ so

$$\frac{dH_1}{du_n} = \overline{\sum_{k,m=-N}^N (V_N)^n_k (V_N)^k_m u_m} = \overline{F^n_1}$$
or

$$\frac{dH_1}{d\overline{u_n}} = F_n^1$$
therefore $F_1 = V_N^2 u = \nabla H_1 $. The function $H_2$ can be rewritten 

$$H_2 = \frac{1}{6} \sum_{k+l+j = 0} \sum_{n,m,y=-N}^N (V_N)^k_n(V_N)^l_m (V_N)^j_y u_n u_m u_y$$
so its complex derivative wrt $u_n$ is : 

$$\frac{dH_2}{du_n} = \frac{1}{2} \sum_{k+l+j = 0} \sum_{m,y=-N}^N (V_N)^k_n(V_N)^l_m (V_N)^j_y u_m u_y$$

$$= \overline{\frac{1}{2} \sum_{k+l+j = 0} \sum_{m,y=-N}^N (V_N)^n_k(V_N)^{-l}_m (V_N)^{-j}_y u_m u_y}$$

$$= \overline{\frac{1}{2} \sum_{k+l+j = 0} (V_N)^n_k(V_N u)_{-l} (V_N u)_{-j} } = \overline{\frac{1}{2} \sum_{k+l+j = 0} (V_N)^n_k(V_N u)^2_{-l-j}} = \overline{ F^n_2}\; .$$

Therfore, the equation \reff{perturbeq} is written : 

$$\partial_t u = -J_N \nabla_u H $$
it has a Hamiltonian form.\end{proof}

Proof of the Proposition \ref{invarn}.

\begin{proof}The equation being Hamiltonian, the Lebesgue measure on $E_0^N$, that is, $\prod da_n db_n$ when $u\in E_0^N$ is written $a_0 c_0 +\sum_{n=1}^N a_n c_n + b_n s_n$ is invariant through its flow $\phi_V^N$, see the proof of proposition \reff{liouville}. 

Then, the measure $\mu_V^N$ is given by

$$d\mu_V^N (u) = d_V^N e^{-\int uW_N u}da_0\prod_{n=1}^N da_n db_n $$
with $d_V^N$ a normalisation factor that depends only on the eigenvalues of $W_N$ and thus does not depend on time. What is more

$$E_V = \int \phi_V^N(t) u_0 W_N \phi_V^N(t) u_0$$
does not either depend on time. So the measure $\mu_V^N$ is invariant through the flow of the equation \reff{perturbeq}.\end{proof}

\begin{definition} Let $\nu_V^N$ the measure on $L^2$ defined as $\mu_V^N \otimes \mu_{M+1}$.\end{definition}

\begin{proposition}The equation 

\begin{equation}\label{perturbinf} \left \lbrace{ \begin{tabular}{ll}
$\partial_t ((1-\Pi_{N}) H^2 (1-\Pi_{N}) + W_N ) u + \partial_x (1-\Pi_N) u + V_n \partial_x \left(V_N u + \frac{(V_N u)^2}{2}\right)=0$ \\
$u|_{t=0} = u_0 \in L^2$ \end{tabular}} \right. \end{equation}
admits a unique global solution $\psi_V^N(t)(u_0)$.

What is more, $\nu_V^N$ is invariant under $\psi_V^N$. \end{proposition}

\begin{proof}If $u$ is decomposed as $u = v+w = \Pi_N u + (1-\Pi_N) u$ then the equation \reff{perturbinf} is equivalent to

$$\partial_t H^2 w + \partial_x w = 0 \mbox{ with } w_{t=0} = (1-\Pi_N) u_0$$
and

$$\partial_t W_N v + V_N \partial_x \left( V_N v + \Pi_N \frac{(V_Nv)^2}{2}\right) \mbox{ with } v_{t=0} = \Pi_N u_0 \; .$$

Hence the flow $\psi_V^N (t)= \phi_V^N (t)+ S(t)$ globally exists and is unique and the measures of the Cartesian products are invariant under the flow, so the measure $\nu_V^N$ is invariant for all measurable sets. \end{proof}

\subsection{Properties of the new measure}

In order to prove the invariance of the perturbed measure under the infinite dimensional perturbed flow, first estimate the measure of some compact sets in $L^2$ (the same as for $\mu$). The proof of the following proposition is greatly ressembling what have been done by Tzvetkov in \cite{limsup}. 

\begin{proposition} Let $U$ be an open set of $L^2$. The measures of $U$ satisfy : 

$$\mu_V (U) \leq \liminf_{N \rightarrow \infty } \nu_V^N (U) \; .$$
\end{proposition}

\begin{proof}Let 

$$\widetilde \varphi_V^N = \varphi_V^N + \varphi - \varphi_0^N \; .$$

This sequence converges in $L^2$ and thus almost surely in $\omega$ toward $\varphi_V$. Indeed,

$$E (||\varphi_V - \widetilde \varphi_V^N ||_{L^2}^2) \leq C \left( E( ||\varphi_V -\varphi_V^N||_{L^2}^2) + E(||\varphi - \varphi_0^N||_{L^2}^2) \right) \rightarrow_{N\rightarrow \infty} 0 \; .$$

Let $Z$ be the set included in $\Omega$ such that $\widetilde \varphi_V^N$ converges.

Let 

$$A = \lbrace \omega\in Z \; | \; \varphi_V (\omega ) \in U \rbrace = (\varphi_V)^{-1}(U)\cap Z$$

$$A^N = \lbrace \omega \in Z \; | \; \widetilde \varphi_V^N (\omega ) \in U \rbrace  = (\widetilde \varphi_V^N)^{-1}(U)\cap Z\; .$$

As $P(Z)=1$,

$$\mu_V(U) = P(A) \mbox{ and } \nu_V^N (U)= P(A^N) \; .$$

If $\omega \in A$, as $U$ is open, there exists $\epsilon >0$ such that $\varphi_V(\omega ) + B_\epsilon \subseteq U$. Then, there exists $N_0$, such that for all $N\geq N_0$,

$$\widetilde \varphi_V^N(\omega) \in \varphi_V (\omega ) + B_\epsilon \subseteq U$$
as the sequence $\widetilde \varphi^N_V$ converges toward $\varphi_V$ in $L^2_x$. So, for all $N\geq N_0$,

$$\omega \in A^N$$

$$\omega \in \liminf A^N$$

$$A\subseteq \liminf A^N \; .$$

Hence, thanks to Fatou lemma,

$$\mu_V (U) = P(A) \leq P(\liminf A^N) \leq \liminf P(A^N) = \liminf \nu_V^N (U) \; .$$

\end{proof}

The same property also holds for the $\mu_V^N$ :

\begin{proposition} Let $U$ be an open set of $L^2$. The measures of $U$ satisfy : 

$$\mu_V (U) \leq \liminf_{N \rightarrow \infty } \mu_V^N (U\cap E_0^N) \; .$$
\end{proposition}

\begin{proof}The proof is very similar to the one above. It uses the convergence of the $\varphi_V^N$ toward $\varphi_V$.\end{proof}

\begin{proposition}\label{compactness} Let $s\in [0,\frac{1}{2}[$. There exists $C,c$ two constants such that for all $R\geq 0$, the measure of the complementary of the closed ball in $H^s$ of radius $R$ satisfies : 

$$\mu_V((B_R^s)^c) \leq C e^{-cR^2}\; .$$
\end{proposition}

\begin{proposition}\label{compactnenn}There exist $C,c$ two constants such that for all $N_0\in \N$ and $R\geq 0$ : 

$$\mu_V( \lbrace u_0 \; | \; ||(1-\Pi_{N_0})u_0||_{L^2} > R \rbrace ) \leq C e^{-c N_0 R^2} \; .$$\end{proposition}

The proofs of the previous propositions are very similar, hence they shall be proved in parallel. For this, the following lemma should prove itself useful.

\begin{lemma}Let $(a_n)_{n\in \Z}$ such that $\sum_n \frac{a_n^2}{1+n^2} < \infty$. Then

$$P(|\alpha_0^N a_0 + \sum_{n=1}^N \alpha_n^N a_n + \beta_n^N a_{-n} | > \lambda ) \leq 2e^{-c_1 \lambda^2/(\sum_n \frac{a_n^2}{1+n^2})} \; .$$\end{lemma}

\begin{proof}If $\sum_n \frac{a_n^2}{1+n^2} = 0$ then the inequality is satisfied. Else, for all $t> 0$,

$$P(\alpha_0^N c_0 + \sum_{n=1}^N \alpha_n^N a_n + \beta_n^N a_{-n} > \lambda) = P(e^{t\alpha_0^N a_0 + \sum_{n=1}^N \alpha_n^N a_n + \beta_n^N a_{-n}} > e^{t\lambda})$$

$$\leq e^{-t\lambda} E (e^{t \alpha_0^N a_0 + \sum_{n=1}^N \alpha_n^N a_n + \beta_n^N a_{-n}})$$

Compute this average.

$$E (e^{t \alpha_0^N a_0 + \sum_{n=1}^N \alpha_n^N a_n + \beta_n^N a_{-n}}) = \int d_N dx^N e^{t\langle a^N , B_N x^N \rangle }e^{-\langle x^N ,x^N\rangle / 2}$$
where $x^N$ is of size $2N+1$. So,

$$E (e^{t \alpha_0^N a_0 + \sum_{n=1}^N \alpha_n^N a_n + \beta_n^N a_{-n}}) = e^{t^2 \langle B_N^* a^N , B_N^* a^N \rangle /2 }\; .$$

Then, $B_N$ corresponds to the operator $\Pi_N (1+V)^{-1/2} H^{-1} \Pi_N$ so, writing $u_0=a_0$ and for all $n> 0$,

$$u_n = \frac{a_n-ia_{-n}}{2} \; , \; u_{-n} = \frac{a_n + ia_{-n} }{2}$$
such that :

$$a_0 + \sum_{n=1}^N a_n c_n + a_{-n}s_n = \sum_n u_n e^{inx}$$
and

$$(1+V)^{-1/2} = \sum_{n} V_n e^{inx} \; ,$$
$B_N^* a^N$ corresponds to

$$\sum_n \frac{ 1_{|n|\leq N }}{\sqrt{1+n^2}}\left( \sum_k V_k  u_{n-k} 1_{|n-k|\leq N}\right)$$

$$\langle B_N^* a^N , B_N^* a^N \rangle = \sum_{n,k_1,k_2} \frac{ 1_{|n|\leq N }}{1+n^2}\overline V_{k_1} V_{k_2}\overline u_{n-k_1} u_{n-k_2} 1_{|n-k_1|\leq N}1_{|n-k_2|\leq N} $$

$$\leq \sum_{k_1,k_2,n} |V_{k_1}|\; |V_{k_2}| \; \frac{|u_{n-k_1}|}{\sqrt{1+n^2}} \frac{|u_{n-k_1}|}{\sqrt{1+n^2}}\frac{|u_{n-k_2}|}{\sqrt{1+n^2}}$$

Then, use that for all $n$ and all $k$

$$\frac{1}{\sqrt{1+n^2}}\leq \frac{\sqrt{2(1+k^2)}}{\sqrt{1+(n-k)^2}}$$

$$\langle B_N^* a^N , B_N^* a^N \rangle \leq \sum_{n,k_1,k_2} \sqrt{2(1+k_1^2)}|V_{k_1}|\sqrt{2(1+k_2^2)} |V_{k_2}| \frac{|u_{n-k_1}|}{\sqrt{1+(n-k_1)^2}}\frac{|u_{n-k_2}|}{\sqrt{1+(n-k_2)^2}}$$

$$\sum_n \frac{|u_{n-k_1}|}{\sqrt{1+(n-k_1)^2}}\frac{|u_{n-k_2}|}{\sqrt{1+(n-k_2)^2}}\leq \sum_n \frac{|u_n|^2}{1+n^2} = \sum_n \frac{a_n^2}{1+n^2}$$

$$\langle B_N^* a^N , B_N^* a^N \rangle \leq \sum_n \frac{c_n^2}{1+n^2} \left( \sum_k |V_k| \sqrt{2(1+k^2)}\right)^2$$

Now, as $V$ is $\mathcal C^2$ and its norm $||V||_\infty \leq 1/2$,  $(1+V)^{-1/2}$ is also $\mathcal C^2$ so the sum : 

$$\sum_k |V_k| \sqrt{2(1+k^2)}\leq \sqrt 2 \left( V_k^2 (1+k^2)^2 \right)^{1/2} \left( \sum_k \frac{1}{1+k^2} \right)^{1/2} \leq C ||V||_\infty$$
converges and is bounded uniformly in $V$.

$$\langle B_N^* a^N , B_N^* a^N \rangle \leq C \sum_n \frac{a_n^2}{1+n^2}$$

$$P(\alpha_0^N a_0 + \sum_{n=1}^N \alpha_n^N a_n + \beta_n^N a_{-n} > \lambda) \leq e^{-t\lambda} e^{Ct^2 \sum \frac{a_n^2}{1+n^2}/ 2}$$

With $t = \frac{\lambda }{C\sum (a_n^2 / (1+n^2)) }$,

$$P(\alpha_0^N a_0 + \sum_{n=1}^N \alpha_n^N a_n + \beta_n^N a_{-n} > \lambda) \leq e^{-\frac{\lambda^2 }{2C\sum (a_n^2/(1+n^2)) }}\; ,$$
and with the same kind of arguments,

$$P(\alpha_0^N a_0 + \sum_{n=1}^N \alpha_n^N a_n + \beta_n^N a_{-n} < - \lambda) \leq e^{-\frac{\lambda^2 }{2C\sum \frac{a_n^2}{1+n^2} }}\; ,$$
so

$$P(|\alpha_0^N a_0 + \sum_{n=1}^N \alpha_n^N a_n + \beta_n^N a_{-n} |> \lambda) \leq 2 e^{-\frac{\lambda^2 }{2C\sum \frac{a_n^2}{1+n^2} }}\; ,$$

\end{proof}

\begin{lemma} There exists $C_1$ such that for all $q\geq 1 $, and all sequence $a_n$ and all $N$,

$$||a_0 \alpha_0^N + \sum_{n=1}^N a_n \alpha_n^N + a_{-n}\beta_n^N ||_{L^q_\omega} \leq \sqrt{C_1 q\sum_n \frac{a_n^2}{1+n^2}} \; .$$
\end{lemma}

\begin{proof} 
$$||a_0 \alpha_0^N + \sum_{n=1}^N a_n \alpha_n^N + a_{-n}\beta_n^N ||_{L^q_\omega}^q = \int q \lambda^{q-1} P(|\alpha_0^N c_0 + \sum_{n=1}^N \alpha_n^N a_n + \beta_n^N a_{-n} |> \lambda)d \lambda$$

$$\leq  \int q \lambda^{q-1} 2 e^{-c_1 \lambda^2/ 2\sum_n \frac{a_n^2}{1+n^2}}$$

By a change a variable $y= \frac{\sqrt{c_1}\lambda}{\sqrt{\sum_n a_n^2/(1+n^2)}}$,

$$\leq 2 \left(\sum_n \frac{(a_n)^2}{c_1(1+n^2)} \right)^{q/2} \int qy^{q-1} e^{-y^2/2}dy$$

$$||a_0 \alpha_0^N + \sum_{n=1}^N a_n \alpha_n^N + a_{-n}\beta_n^N ||_{L^q_\omega}\leq \left(4\sum_n \frac{(a_n)^2}{c_1(1+n^2)} \right)^{1/2}\left( \int qy^{q-1} e^{-y^2/2}dy \right)^{1/q}$$

For all $q\in [1,3]$,

$$\frac{2}{\sqrt{c_1}} \left( \int qy^{q-1} e^{-y^2/2}dy \right)^{1/q} \leq \frac{2}{\sqrt{c_1}} 6\sqrt{2\pi} = C'_1$$

If $q\geq 3$

$$\int qy^{q-1} e^{-y^2/2}dy  = q(q-2) \int y^{q-3} e^{-y^2/2}$$

With $K = \lceil \frac{q-3}{2} \rceil $, and by recurrence, 

$$\int qy^{q-1} e^{-y^2/2}dy = \prod_{k=0}^{K-1} (q-2k) (q-2K)\int y^{q-2K-1} e^{-y^2/2} dy$$

$$\leq q^K C_1\leq q^{q/2} C'_1$$
as $q-2K \in [1,3]$.

$$||a_0 \alpha_0^N + \sum_{n=1}^N a_n \alpha_n^N + a_{-n}\beta_n^N ||_{L^q_\omega}\leq \left(C_1 q \sum_n \frac{(a_n)^2}{(1+n^2)} \right)^{1/2}$$  \end{proof}

Proof of the propositions \reff{compactness},\reff{compactnenn}.

\begin{proof} For all $q\geq 2$,

$$\mu_V^N ( (B_R^s)^c \cap E_0^N) = P( ||\varphi_V^N (\omega)||_{H^s} > R) = P(||H^s \varphi_V^N (\omega)||_{L^2} > R)$$

$$ = P(||H^s \varphi_V^N (\omega)||_{L^2}^q > R^q) \leq R^{-q} E(||H^s \varphi_V^N (\omega)||_{L^2}^q) = R^{-q}||H^s \varphi_V^N (\omega)||_{L^q_\omega,L^2_x}^q$$

$$\leq R^{-q}||H^s \varphi_V^N||_{L^2_x,L^q_\omega}^q$$
thanks to Minkowski inequality.

Similarly,

$$\mu_V^N  (\lbrace u_0 \; |\; ||(1-\Pi_{N_0})u_0 ||_{L^2} > R \rbrace ) \leq R^{-q} ||(1-\Pi_{N_0})\varphi_V^N ||_{L^2_x,L^q_\omega}^q\; .$$

$$H^s \varphi_V^N  = \alpha_0^N c_0(x) + \sum_n \left( (1+n^2)^{s/2}c_n(x) \alpha_n^N + (1+n^2)^{s/2}s_n(x) \beta_n^N\right) $$
so

$$||H^s \varphi_V^N (x)||_{L^q_\omega}\leq \sqrt{C_1 q \left( c_0(x)^2 + \sum_n \frac{c_n(x)^2 + s_n(x)^2}{(1+n^2)^{1-s}}\right) }$$

$$||H^s \varphi_V^N (x)||_{L^2_x,L^q_\omega} \leq \sqrt{C_1 q \left( 1 + \sum_n \frac{2}{(1+n^2)^{1-s}}\right)}$$
and

$$||(1-\Pi_{N_0})\varphi_V^N ||_{L^2_x,L^q_\omega} \leq \sqrt{C_1q \sum_{n> N_0} \frac{2}{1+n^2}} \leq \sqrt{\frac{2C_1}{N_0}}\; .$$

As $s< \frac{1}{2}$, $1-s> \frac{1}{2}$ so the series converges.

$$||H^s \varphi_V^N (x)||_{L^2_x,L^q_\omega} \leq \sqrt{C_2 q }$$
where $C_2$ depends on $s$ but not on $N$.

$$\mu_V^N ( (B_R^s)^c \cap E_0^N) \leq \left( \frac{C_2 q}{R^2} \right)^{q/2}$$.

Also,

$$\mu_V^N  (\lbrace u_0 \; |\; ||(1-\Pi_{N_0})u_0 ||_{L^2} > R \rbrace ) \leq \left( \frac{2C_1q}{N_0 R^2} \right)^{q/2} \; .$$

For $R\leq \sqrt{2eC_2}$, see that with $c = \frac{1}{2eC_2}$,

$$\mu_V^N ( (B_R^s)^c \cap E_0^N) \leq 1 \leq e^{1/2} e^{-cR^2} = C e^{-cR^2}$$
and if $R\geq \sqrt{2eC_2}$, by replacing $q$ with $ \frac{R^2}{eC_2} \geq 2$,

$$\mu_V^N ( (B_R^s)^c \cap E_0^N) \leq e^{-q/2} = e^{-R^2/ (2eC_2)}= e^{-cR^2}\leq Ce^{-cR^2}$$
hence, for all $R$ and all $N$

$$\mu_V^N ( (B_R^s)^c \cap E_0^N) \leq Ce^{-cR^2} \; .$$

With the same kind of arguments,

$$\mu_V^N (\lbrace u_0 \; |\; ||(1-\Pi_{N_0})u_0 ||_{L^2} > R \rbrace ) \leq C^{e^{-c N_0R^2}} \; .$$

Now see that $B_R^s$ is closed in $L^2$ so $(B_R^s)^c$ and $\lbrace u_0 \; |\; ||(1-\Pi_{N_0})u_0 ||_{L^2} > R \rbrace $ are open in $L^2$,

$$\mu_V ((B_R^s)^c)) \leq \liminf \mu_V^N ( (B_R^s)^c \cap E_0^N) \leq Ce^{-cR^2} \; $$
and

$$\mu_V (\lbrace u_0 \; |\; ||(1-\Pi_{N_0})u_0 ||_{L^2} > R \rbrace ) \leq Ce^{-cN_0R^2} \; .$$
\end{proof}

\section{Convergence of the flows}

To study the invariance of the perturbed measure, the property of local uniform convergence of the ``finite dimensional" flows on compacts is needed. 

\subsection{Properties of the finite dimensional operators}

First, investigate on the different operators involved.

\begin{lemma}\label{improv} Let $s< \frac{1}{2}$ and $s_1 > 0$. Let $K$ be the operator defined as : 

$$K = (1-\partial_x^2)^{-1} \partial_x\; .$$

There exists $C$ such that for all $u,v$ in $L^2$ and $g$ a linear operator defined on $L^1$ such that

$$||g|| := \sup_{n\in \Z} \sum_{m\in \Z} |g_m^n| = \sup_{n\in \Z}\sum_m |\frac{1}{2\pi} \int e^{-inx}g(e^{imx}) |$$
is finite, and for all $N\geq 1$,

\begin{enumerate}
\item $||K(g (uv)||_{H^s} \leq C ||g|| \; ||u||_{L^2}||v||_{L^2} $,
\item if $u,v$ are in $H^{s_1}$, $||K(g (1-\Pi_N) uv) \leq C N^{-s_1} ||g||\; ||u||_{H^{s_1}}||v||_{H^{s_1}}$ .
\end{enumerate}\end{lemma}

\begin{proof}Write

$$u(x) = \frac{1}{\sqrt{2\pi}} \sum_{k \in \Z} u_k e^{ikx}$$
and

$$v(x) = \frac{1}{\sqrt{2\pi}} \sum_{k \in \Z} v_k e^{ikx}$$
their Fourier series. As $uv $ belongs to $L^1$,

$$uv(x) = \frac{1}{2\pi} \sum_{k\in \Z} \left( \sum_{l}u_{k-l}v_l \right) e^{ikx}$$

$$g(uv) (x) = \frac{1}{\sqrt{2\pi}}\sum_{m\in \Z} e^{imx} \sum_{k\in \Z} g^m_k \sum_{l\in \Z} u_{k-l}v_l\; .$$

Now, give an upper bound of the $n$-th Fourier coefficient $c_n$ of $g(uv)$,

Reversing the sums over $k$ and $l$ gives : 

$$| c_n| \leq \sum_l |v_l| \sum_{k} |g_{k}^n|\; |u_{k-l}| $$
and by a Cauchy-Schwartz inequality

$$|c_n| \leq ||v||_{L^2} \left( \sum_l \left( \sum_k |g_{k}^n|\; |u_{k-l}|\right)^2 \right)^{1/2}\; .$$

Now

$$\sum_l \left( \sum_k |g_{k}^n|\; |u_{k-l}|\right)^2 = \sum_{l,k_1,k_2} |g_{k_1}^n|\; |g_{k_2}^n| \; |u_{k_1-l}||u_{k_2-l}|$$

$$=\sum_{k_1,k_2}|g_{k_1}^n|\; |g_{k_2}^n| \sum_l |u_{k_1-l}||u_{k_2-l}|\; .$$

Use Cauchy-Schwartz inequality a second time to get : 

$$\sum_l |u_{k_1-l}||u_{k_2-l}| \leq ||u||_{L^2}^2$$
so

$$|c_n| \leq ||u||_{L^2}||v||_{L^2}\sqrt{\sum_{k_1,k_2}|g_{k_1}^n||g_{k_2}^n| }$$
but

$$\sum_{k_1,k_2}|g_{k_1}^n||g_{k_2}^n| = \left( \sum_k |g_{k}^n| \right)^2 \leq ||g||^2$$
hence the result. Indeed,

$$||K(g (uv))||_{H^s} =\left( \sum_n \frac{n^2}{(1+n^2)^{2-s}} |c_n|^2 \right)^{1/2} \leq ||u||_{L^2}||v||_{L^2}||g|| \sqrt{\sum_k \frac{k^2}{(1+k^2)^{2-s}}}$$
and the series converges.

For the third one, $c_k$ the $k$-th Fourier coefficient of $g((1-\Pi_N)(uv))$ is given by

$$c_k = \sum_{l,m} g_{m}^k 1_{|m|> N} u_{m-l} v_{l} \; .$$

See that if $|m|> N$ then or $|m-l|> N/2$, or $|l|> N/2$, so $1_{|m|> N} \leq 1_{|m-l|>N/2} + 1_{|l|>N/2}$ and

$$|c_k| \leq \sum_{l,m} |g_{m}^k| 1_{|m-l|> N/2} |u_{m-l}| |v_{l}| + \sum_{l,m} |g_{m}^k|  |u_{m-l}| 1_{|l|> N/2}|v_{l}|\; .$$

Estimate the second sum as the two of them are symmetrical.

$$\sum_{l,m} |g_{m}^k|  |u_{m-l}| 1_{|l|> N/2}|v_{l}|\leq ||(1-\Pi_{N/2} ) v||_{L^2}\left(\sum_l \left(  \sum_m |g_{m}^k|  |u_{m-l}| \right)^2\right)^{1/2}$$

$$\leq (N/2)^{-s_1}||v||_{H^{s_1}}||u||_{L^2}||g|| \leq (N/2)^{-s_1}||v||_{H^{s_1}}||u||_{H^{s_1}}||g|| \; .$$

Hence the result.

\end{proof}

\begin{definition}Let $W = \sqrt{1+V}H^{2} \sqrt{1+V}$ and $D = W^{-1}(1+V)^{1/2}\partial_x (1+V)^{1/2} = (1+V)^{-1/2}K (1+V)^{1/2}$. Call also $D_N=V_N^{-1} K V_N$ and $K_N = \Pi_N K \Pi_N$.\end{definition}

\begin{lemma}Let $s< 1/2$. The operators $D$ and $D_N$ are defined and continuous from $L^2$ to $L^2$ and there exists $C$ such that for all $V$, $N$ : 
\begin{enumerate}
\item for all $u,v \in L^2$, $||K_N(uv)||_{H^s} \leq C||u||_{L^2}||v||_{L^2}$,
\item for all $u,v \in L^2$ , $||D_N V_N^{-1}(\Pi_N (uv)) ||_{H^s}\leq C ||u||_{L^2}\; ||v||_{L^2} $,
\item for all $u \in H^s$, $t\in \R$, $||e^{-tD_N}u||_{L^2} \leq e^{c|t|} ||u||_{H^s}$,
\item for all $u,v \in L^2$ , $||D((1+V)^{-1/2}uv) ||_{H^s}\leq C ||u||_{L^2}\; ||v||_{L^2} $,
\item for all $u \in H^s$, $t\in \R$, $||e^{-tD}u||_{H^s} \leq e^{c|t|} ||u||_{H^s}$. 
\end{enumerate}
\end{lemma}

\begin{proof}The first, second and fourth inequalities are consequences of the previous lemma with $g=\Pi_N$ or $g = Id_{L^2}$, using the fact that the norms of the operators $V_N$, $V_N^{-1}$, $(1+V)^\alpha$ are uniformly bounded in $V$ and $N$.

To obtain the third or the fifth one, observe that

$$D_N = V_N^{-1} K V_N \mbox{ and } D= (1+V)^{-1/2} K (1+V)^{1/2}$$
so the norm of $D_N$ as an operator is uniformly bounded in $V$ and $N$. 

The main problem is that $D_N$ or $D$ are not antisymmetric, so $e^{-tD_N}$ or $e^{-tD}$ can not be an isometry. Nevertheless setting $f(t) = ||e^{-tD_N }u||_{H^s}$ and using Gronwall lemma, as

$$f(t) \leq \int_{t'=0}^t ||D_N e^{-t'D_N}u ||_{H^s} dt' \leq c \int_{0}^t f(t') dt'$$

$$f(t) \leq f(0) e^{c|t|} = ||u||_{H^s} e^{c|t|} \; .$$
\end{proof}
 
\subsection{Local existence and convergence of the finite dimensional perturbed flows}
 
Show now the local well posedness of the perturbed equations and the uniform convergence of the $2N+1$ dimensional solutions toward the infinite dimensional one on compact sets.

\begin{definition}Let $u_0 \in L^2$ and $A_V^N$ and $A_V$ be defined on $X^0_T$ as

$$A_V^N (u) = e^{-t(1-\Pi_N )K }u_0 + e^{-tD_N}u_0 + \int_{0}^t e^{-(t-s)D_N}D_N V_N^{-1}\Pi_N \frac{(V_N\Pi_N u(s))^2}{2}ds$$
and

$$A_V (u) = e^{-tD}u_0 + \int_{0}^t e^{-(t-s)D}D (1+V)^{-1/2}\frac{((1+V)^{1/2} u(s))^2}{2}ds\; .$$
\end{definition}

\begin{proposition}There exists $C$ independent from $u_0\in H^s$ and $N$ such that for all $T\leq 1$, $u,v \in X^s_T$, $s<1/2$,
\begin{enumerate}
\item $||A_V^N (u)||_{X^s_T} \leq  C \left( ||u_0||_{H^s} + T ||u||_{X^s_T}^2 \right)$,
\item $||A_V (u)||_{X^s_T} \leq  C \left( ||u_0||_{H^s} + T ||u||_{X^s_T}^2 \right)$,
\item $||A_V^N(u)-A_V^N(v)||_{X^s_T} \leq C T \left( ||u||_{X^s_T}+ ||v||_{X^s_T}\right) ||u-v||_{X^s_T}$,
\item $||A_V(u)-A_V(v)||_{X^s_T} \leq C T \left( ||u||_{X^s_T}+ ||v||_{X^s_T}\right) ||u-v||_{X^s_T}$.
\end{enumerate}
\end{proposition}

\begin{proof}Write $A_V^N(u) = I +II+III$ with

$$I = e^{-t(1-\Pi_N )K }u_0 \; , \; II = e^{-tD_N}u_0$$
and

$$III = \int_{0}^t e^{-(t-s)D_N}D_N V_N^{-1}\Pi_N \frac{(V_N \Pi_N u(s))^2}{2}ds \; .$$

$$||I||_{X^s_T} \leq ||(1-\Pi_N)u_0||_{H^s} \leq ||u_0||_{H^s}$$

$$||II||_{X^s_T} \leq e^{cT} ||u_0||_{H^s} \leq C ||u_0||_{H^s}$$

$$||III||_{H^s} = \left( \int_x | H^s \int_{0}^t e^{-(t-s)D_N}D_N V_N^{-1} \Pi_N \frac{(V_N\Pi_N u(s))^2}{2}ds |^2 \right)^{1/2} $$

$$\leq \left( \int_x T\int_{0}^t |H^s e^{-(t-s)D_N}D_N V_N^{-1}\frac{(V_N\Pi_N u(s))^2}{2}|^2 ds\right)^{1/2}$$

$$\leq T^{1/2}\left( \int_{0}^t ||e^{-(t-s)D_N}D_N V_N^{-1}\Pi_N \frac{(V_N \Pi_N u(s))^2}{2}||_{H^s}^2 \right)^{1/2}$$

But, as

$$||e^{-(t-s)D_N}D_N V_N^{-1}\Pi_N \frac{(V_N \Pi_N u(s))^2}{2}||_{H^s} \leq C ||D_N V_N^{-1}\Pi_N \frac{(V_N \Pi_N u(s))^2}{2}||_{H^s}$$

$$\leq C ||V_N \Pi_N u(s)||_{L^2}^2\leq C ||u||_{X^s_T}^2$$

$$||III||_{H^s} \leq C T ||u||_{X^s_T}^2 $$
so

$$||A_V^N(u)||_{X^s_T} \leq C(||u_0||_{H^s} + T ||u||_{X^s_T}^2 )\; .$$

The same proof holds for (2).

For (3) and (4), compute $A_V^N(u)-A_V^N(v)$ : 

$$A_V^N(u)-A_V^N(v) = \int_{0}^t e^{-(t-s)D_N}D_N V_N^{-1}\Pi_N \frac{(V_N \Pi_N u(s))^2- (V_N \Pi_N v(s))^2}{2}ds$$

$$=\int_{0}^t e^{-(t-s)D_N}D_N V_N^{-1}\Pi_N \frac{(V_N \Pi_N (u+v))(V_N \Pi_N (u-v))}{2}ds $$

Hence, with the same computation as before for $III$: 

$$||A_V^N(u)-A_V^N(v)||_{X^s_T} \leq CT ||u+v||_{X_T^s}||u-v||_{X^s_T} \leq CT \left(||u||_{X_T^s}+||v||_{X_T^s}\right)||u-v||_{X_T^s}\; .$$
\end{proof}

\begin{proposition}Let $R> 0$, and $0\leq s< 1/2$, there exists $C_s$ such that for all $u_0 \in B_R^s$, the flows $\psi_V^N$ of \reff{perturbinf} and $\psi_V$ of :

\begin{equation}\label{perturbe} \left \lbrace{ \begin{tabular}{ll}
$\partial_t \left( (1+V)^{1/2} (1-\partial_x^2)(1+V)^{1/2} \right) u+ (1+V)^{1/2}\partial_x \left( (1+V)^{1/2} u + \frac{((1+V)^{1/2}u)^2}{2} \right) = 0$\\
$u|_{t=0} = u_0$ \end{tabular}} \right. \end{equation}
are defined for $t\in [-T,T]$ with $T = \frac{1}{8C_s^2 R}$ and satisfy

$$||\psi_V(u_0) ||_{X^s_T}\; , \; ||\psi_V^N(u_0) ||_{X^s_T} \leq 2C_s R \; .$$
\end{proposition}

\begin{proof}For all $u_0 \in B_R^s$ and $u,v \in X_T^s$ with $||u||_{X^s_T},||v||_{X^s_T} \leq 2C_sR$,

$$||A_V(u)||_{X^s_T},||A_V^N(u)||_{X^s_T}\leq C_s (R + T(2C_s R)^2)\leq 2C_s R$$
and  

$$||A_V(u)-A_V(v)||_{X^s_T}\; ,||A_V^N(u)-A_V^N(v)||_{X^s_T} \leq C_s (4C_s R) T||u-v||_{X^s_T} = \frac{1}{2} ||u-v||_{X^s_T}$$
so both $A_V^N$ and $A_V$ have a unique fix point in the ball of radius $2C_sR$ in $X_T^s$. \end{proof}

As for BBM, there is a property of uniform convergence on compacts : 

\begin{lemma}Let $\epsilon > 0$ and $R>0$. There exists $N_0\in \N$ such that for all $u_0 \in B_R^s$, for all $N\geq N_0$, and $T \leq \frac{1}{C_2 R}$,

$$||\psi_V(u_0) -\psi_V^N(u_0)||_{X_T^0}\leq \epsilon \; .$$\end{lemma}

\begin{proof} Let $u = \psi_V(u_0)$ and $u^N= \psi_V^N (u_0)$,

$$||u-u^N||_{X_T^s} = ||A_V(u) - A_V^N(u^N) ||_{X_T^0} = ||I+II+III+IV+V||_{X_T^0}$$
with

$$I = (1-\Pi_N)S(t) u_0$$

$$II = (e^{-tD}-e^{-tD_N})u_0$$

$$III = \int_{0}^t \left( e^{-(t-s)D}D(1+V)^{-1/2} - e^{-(t-s)D_N}D_NV_N^{-1}\Pi_N \right) \frac{((1+V)^{1/2}u)^2(s)}{2}ds \; ,$$

$$IV = \int_{0}^t  e^{-(t-s)D_N}D_N V_N \Pi_N \left( \frac{((1+V)^{1/2}u)^2(s)}{2} -\frac{(V_N \Pi_N u)^2(s)}{2}\right) ds $$
and 

$$V = \int_{0}^t e^{-(t-s)D_N}D_N V_N \Pi_N \left( \frac{(V_N u)^2(s)}{2} -\frac{(V_N u^N)^2(s)}{2}\right) ds\; .$$

First, remark that

$$||(D-D_N)(uv) ||_{L^2} \leq C_s N^{-s} ||u||_{H^s}||v||_{H^s}\; .$$

Indeed, $D-D_N$ can be written,

$$D-D_N = \left( (1+V)^{-1/2}-V_N^{-1} \right) K (1+V)^{1/2} + V_N^{-1} K \left( (1-\Pi_N) (1+V)^{1/2} \right) $$
\hspace{1cm} $$+ V_N^{-1} K (1-\Pi_N) (1+V)^{1/2} (1-\Pi_N)$$

As operators, the multiplication by $(1+V)^{1/2}$  and $V_N$ are quite close. Computing the difference between their inverse gives

$$||(1+V)^{-1/2} - V_N^{-1} ||_{H^s \rightarrow L^2} \leq ||1-\Pi_N||_{H^s \rightarrow L^2} \left( 1+ \sum_k k ||(1+V)^{-1/2} -1||_{H^s\rightarrow H^s}^k \right)\leq C N^{-s}$$
as $ ||(1+V)^{-1/2} -1||_{H^s\rightarrow H^s} \leq ||(1+V)^{-1/2} -1||_{L^\infty} \leq \frac{1}{\sqrt 2}$.

Therefore, it appears that :

$$||\left( (1+V)^{-1/2}-V_N^{-1} \right) K (1+V)^{1/2} (uv)|| \leq CN^{-s} ||K(1+V)^{1/2}(uv)||_{H^s}$$
and as $\sum_m |(1+V)^{1/2}|_m^n = \sum_m |(1+V)^{1/2}|_{n-m}$ is the sum of the Fourier coefficients of $\sqrt{1+V}$ which is $\mathcal C^2$ with a uniform bound on its second derivative, it comes that

$$||K(1+V)^{1/2} (uv)||_{H^s} \leq C_s ||u||_{L^2} ||v||_{L^2 }\leq C_s ||u||_{H^s} ||v||_{H^s }$$
with $C_s$ a constant that does neither depend on $N$ nor on $V$ : 

$$||\left( (1+V)^{-1/2}-V_N^{-1} \right) K (1+V)^{1/2} (uv)|| \leq C_sN^{-s}||u||_{H^s}||v||_{H^s}\; .$$

As the norm of $V_N$ as an operator from $H^s$ to $H^s$ is bounded uniformly in $V$ and $K$ and $\Pi_N$ commute, 

$$||V_N^{-1} K \left( (1-\Pi_N) (1+V)^{1/2} \right)||_{H^s} \leq C|| K \left( (1-\Pi_N) (1+V)^{1/2} \right)(uv)||_{H^s}$$

$$\leq C_sN^{-s} ||(1+V)^{1/2}u||_{L^2}||v||_{L^2} \leq C_s N^{-s} ||u||_{H^s}||v||_{H^s}\; .$$

The same goes for $V_N^{-1} K (1-\Pi_N) (1+V)^{1/2} (1-\Pi_N)$ as the sum of the Fourier coeffocients of $(1+V)^{1/2}$ are uniformly bounded in $V$ so 

$$||(D-D_N)(uv) ||_{L^2} \leq C_s N^{-s} ||u||_{H^s}||v||_{H^s}\; .$$

After what the $I$ and $II$ are less than $C_sN^{-s}R$, $III$ and $IV$ are less than $C_sN^{-s}R^2$ and $V$ by

$$V \leq C_sT R ||u-u^N||_{X_T^0}\; .$$

Indeed,

$$||II||_{X^0_T} \leq T ||(D-D_N)e^{-tD}u_0||_{L^2} + T  ||D_N II ||_{X^0_T}$$

$$ \leq CTN^{-s} ||u_0||_{H^s} +CT ||II||_{X^0_T}$$
so for $T$ small enough, $T\leq \frac{1}{2C_s R}$, the uniform convergence is satisfied.
\end{proof}

\subsection{Invariance of the perturbed measure under the perturbed flow}

Show now that the perturbed measure is invariant trough the perturbed flow. For that, the techniques used are basically the same as in the first section, in particular regarding the local invariance.

\begin{lemma}\label{locinvp} Let $s\in ]0,\frac{1}{2}[$ and $R> 0$. Let $A$ be a measurable set of $L^2$ included in $B_R^s$. Let $T= \frac{1}{C R}$ with $C$ depending on $s$ big enough , for all $t\in [-T,T]$,

$$\mu(\psi_V (t)A)= \mu(A) \; .$$\end{lemma}

\begin{proof}Use the invariance of $\nu_V^N$ through $\psi_V^N$, the uniform convergence of $\psi_V^N$ toward $\psi_V$, the uniform continuity of the flows $\psi_V^N$ and $\psi_V$ and the fact that for all open set $U$ : 

$$ \mu_V (U) \leq \liminf \nu_V^N (U) \; .$$

Indeed, if $A$ is closed, as $A+ B'_\epsilon$ is open (the $B'$ denotes the open ball in $L^2$) : 

$$\mu_V(\psi_V(t)(A+B_\epsilon ')) \leq \mu_V (\psi_V(t)(A) + B'_{C\epsilon}) \leq \liminf \nu_V^N(\psi_V(t) A + B'_{C\epsilon})$$

Then, use that $\psi_V(t) A \subseteq \psi_V^N(t) A + B'_\epsilon$ above a certain $N$.

$$\mu_V(\psi_V(t)(A+B_\epsilon ')) \leq \liminf \nu_V^N(\psi_V^N (t)A B'_{C\epsilon} )\leq \limsup \nu_V^N (\psi_V^N(t) A + B_{C\epsilon})$$
and as the flow is locally continuous in $L^2$,

$$\mu_V(\psi_V(t)(A+B_\epsilon ')) \leq \limsup \nu_V^N(\psi_V^N (t) (A+B_{C'\epsilon})) $$
and as $\nu_V^N$ invariant through $\psi_V^N$,

$$\mu_V(\psi_V(t)(A+B_\epsilon ')) \leq \limsup \nu_V^N (A+B_{C'\epsilon}) \leq \mu_V (A+B_{C\epsilon})$$
and by DCT when $\epsilon $ goes to $0$,

$$\mu_V (\psi_V(t) (A)) \leq \mu_V(A) \; .$$

For the reverse inequality, consider that above a certain $N$, $\psi_V^N(t) A \subseteq \psi_V(t)A + B'_{\epsilon}$, so

$$\mu_V (A+ B'_{\epsilon}) \leq \liminf \nu_V^N(A+B'_\epsilon) \leq \liminf \nu_V^N (\psi_V^N (t) A + B'_{C\epsilon})$$

$$\mu_V (A+ B'_{\epsilon}) \leq \liminf \nu_V^N (\psi_V(t)A + B'_{C'\epsilon} ) \leq \limsup \nu_V^N(\psi_V (t)A + B_{C'\epsilon})\leq \mu_V (\psi_V (t)A + B_{C'\epsilon})$$
and by DCT when $\epsilon$ goes to $0$, as $A$ is closed and thus $\psi(t)A = (\psi(-t))^{-1}A$ is closed too: 

$$\mu_V(A) \leq \mu_V(\psi(t) A)\; .$$

Hence the lemma is true for closed sets. For all the measurable sets, see that $B_R^s$ is closed in $L^2$ so the property passes to the complementary and the countable unions thanks to the uniqueness of the local flow.\end{proof}

Then, define the sets where the solution exists globally in time in quite the same way as in the first section. 

\begin{definition}Let $R>1$ and $R_n = \sqrt{n+1} R$ for $n\geq 0$ and $t_n = \frac{1}{3 C_s \sqrt n R}$ for $n\geq 1$, $T_n = \sum_{k = 1}^n t_n$. Call 

$$A_{V,n}^N (R) = \lbrace \varphi_V(\omega) | \varphi_V^N(\omega ) \in \phi_V^N(T_n)^{-1}(B_{R_n}^s) \rbrace$$
$$A_{V,-n}^N (R) = \lbrace \varphi_V(\omega) | \varphi_V^N(\omega ) \in \phi_V^N(-T_n)^{-1}(B_{R_n}^s) \rbrace $$
and then

$$A_{V}^N (R) = \bigcap_{n\in \Z}  A_{V,n}^N(R)$$

$$A_V(R) = \limsup_{N\rightarrow \infty} A^N_{V}(R)$$
and even

$$A_V = \bigcup_{M\geq 2} A_V(M) \; .$$\end{definition}

\begin{proposition} The set $A_V(R)$ is such that its complementary satisfies :

$$\mu_V (A_V(R)^c ) \leq C e^{-c R^2}$$
and thus

$$\mu_V (A_V^c) = 0\; .$$\end{proposition}

\begin{proof}Consider the sets restricted to the $\omega$ such that $\varphi_V^N $ converges in $H^s$, given that the sequence converges in $L^2(\Omega, H^s)$ and thus almost surely.

It appears that 

$$\mu_V (A_{V,n}^N(R)^c) = P(\varphi_V^N(\omega )\notin \phi_V^N(T_n)^{-1}(B_{R_n}^s))$$ 

$$\mu_V (A_{V,n}^N(R)^c) = \mu_V^N (\left(\phi_V^N(T_n)^{-1}(B_{R_n}^s) \right)^c)$$

$$\mu_V (A_{V,n}^N(R)^c) \leq \mu_V^N (\phi_V^N(T_n)^{-1}(B_{R_n}^s)^c)$$
and $\mu_V^N$ is invariant through the flow $\phi_V^N$ so

$$\mu_V (A_{V,n}^N(R)^c) \leq 2\mu_V^N (B_{R_n}^s) \leq Ce^{-cR_n^2} \leq Ce^{-c(n+1)R^2}\; .$$

Then,
$$\mu_V(A_{V,n}^c) \leq \liminf \mu_V(A_{V,n}^N(R)^c) \leq Ce^{-c(n+1)R^2}$$

$$\mu_V (A_{V}(R)^c) \leq C \sum_{n\geq 1}2e^{-cnR^2} \leq C'e^{-cR^2}$$
and

$$\mu_V(A_V(R)^c) \leq \liminf \mu_V (A_V^N(R)) \leq Ce^{-cR^2}$$

$$\mu_V (A_V^c) = 0 \; . $$ \end{proof}

\begin{proposition} The flow $\psi_V$ is unique and globally defined as long as the initial data is taken in $A_V$ and the measure $\mu_V$ is invariant through this flow. \end{proposition}

\begin{proof} As it happens, the proof is roughly the same as in \reff{blablav}. It uses however the fact that $ \varphi_V^N+ \varphi_{N+1} $ converges almost surely toward $\varphi_V$. Indeed, to study the convergence in $H^s$ at the times $T_n \rightarrow \infty$, see that

$$\psi_V (T_n) (\varphi_V (\omega))$$
is the $H^s$ limit when $N\rightarrow \infty$ of

$$\phi_V^N (T_n) ( \varphi_V^N (\omega) )$$
as

$$\psi_V (T_n) (\varphi_V (\omega)) - \phi_V^N (T_n) ( \varphi_V^N (\omega) ) = \psi_V (T_n) (\varphi_V (\omega)) - \psi_V^N(\varphi_V(\omega)) +\psi_V^N(\varphi_V(\omega)) - \phi_V^N (T_n) ( \varphi_V^N (\omega) )$$

$$ =  \psi_V (T_n) (\varphi_V (\omega)) - \psi_V^N(\varphi_V(\omega)) +\psi_V^N(\varphi_V(\omega))- \psi_V^N (\varphi_V^N(\omega)+\varphi_{N+1}(\omega)) + S(t) \varphi_{N+1}(\omega)$$
and as $\psi_V^N(T_n)$ is continuous the sequence converges in $H^s$. \end{proof}

\section{Evolution of characteristic functionals}

The statistics $\mu_V$ are not too much changed by the flow of the original BBM equation. Though, to investigate about those changes, build the characteristic functional of $\psi(t)(\mu_V)$. Estimations on the characteristic functionals seem relevant, in the sense that they contain all the information about the image measure, and they give precise estimates regarding the small parameter $V$.

\subsection{Definition of the generating functionals}

Introduce now the definition of the generating functionals. 

\begin{definition} Let $\lambda \in L^2$  and let $\langle \lambda, u \rangle$ the scalar product in $L^2$ of $u$ and $\lambda$. Call then $Z_V(\lambda)$ the quantity : 
 
$$Z_V(\lambda ) = E_V \left( e^{i\langle\lambda , u \rangle }\right)$$
where $E_V$ denotes the average over $\mu_V$.\end{definition}

\begin{remark} This functional is the characteristic function of $\mu_V$.

When $V$ is equal to $0$, this functional is equal to 

$$Z_{V=0} (\lambda) = e^{-||\lambda||_{H^{-1}}^2/2}$$
adopting the convention 

$$||\lambda ||_{H^{-1}} = \left(\frac{1}{2\pi}\langle \lambda,c_0\rangle^2 + \frac{1}{2\pi}\sum_{n\geq 1} \frac{1}{1+n^2}(\langle \lambda, c_n\rangle^2 + \langle \lambda , s_n\rangle^2) \right)^{1/2} \; .$$

When $V$ is different from $0$,

$$Z_V(\lambda) = e^{-||(1+V)^{-1/2}\lambda||_{H^{-1}}}\; .$$\end{remark}

Introduce now the generating functional monitoring the behaviour of the BBM flow.

\begin{definition} Let $Z_V(t,\lambda)$ be the quantity : 

$$Z_V(t,\lambda) = E_V( e^{i\langle \lambda , \psi (t)u\rangle}) \; .$$\end{definition}

\begin{remark}First, see that if $\psi$ is replaced by $\psi_V$, this quantity remains the same in time, as $\mu_V$ is invariant through $\psi_V$ and $\psi_V(t)$ is almost surely defined.

Then, it is sufficient to study the interaction between the different waves since the covariance between two modes is given by the behaviour of $Z$ as :

$$E_V(\langle \psi_V(t)u_0,c_n\rangle \langle \psi_V(t)u_0,c_n\rangle ) = - D^2 Z |_{\lambda = 0}(c_n)(c_m) \; .$$
where the right hand term is the second order differential of $Z$ at the point $\lambda = 0$ under the directions $c_n$ and $c_m$. And those quantities are well defined.\end{remark}

\subsection{Closeness of the flows}

First, prove the global existence of the BBM and the perturbed flow, as in \cite{globpos}, along with some useful estimates.

\begin{definition}\label{bign} For all $u_0 \in L^2$ and $T \in \R$ call 

$$N(u_0, T) = \min \lbrace N \in \N \; |\; ||(1-\Pi_N) u_0 ||_{L^2} \leq \frac{1}{C(1+|T|)}\rbrace$$
where $C$ is the constant involved in the $L^2$ local well posedness of the BBM and the perturbed flow.\end{definition}

\begin{proposition} Let $s \in ]0,\frac{1}{2}[$ and $\sigma \in ]\frac{1}{2}, 1]$. There exists $C$ such that for all $u_0\in H^s$, the flows $\psi$ and $\psi_V$ are globally defined in $L^2$, and for all $T\in \R$,

$$||\psi(t) u_0 ||_{L^2}\; , \; ||\psi_V(t)u_0||_{L^2} \leq C+ C N(u_0,T)^{(1+\sigma -2s)/2}||u_0||_{H^s}\; .$$
\end{proposition}

\begin{proof}Fix $T\in \R$ and let $v_0 = (1-\Pi_{N_0}) u_0$ and $w_0 = \Pi_{N_0}u_0$. Thanks to LWP, $\psi(t)v_0$ and $\psi_V(t)v_0$ are defined in $[-T,T]$ in $L^2$ and they satisfy

$$||\psi(t)v_0||_{X^0_T}\; ,\; ||\psi_V(t)v_0||_{X^0_T} \leq C ||v_0||_{L^2} \leq \frac{C}{1+|T|}\; .$$

Call $v = \psi(t)v_0$ and $v_V = \psi_V(t) v_0$. Consider now the equations

\begin{equation}\label{changed} \partial_t (1-\partial_x^2) w = - \partial_x \left( w + vw + \frac{w^2}{2}\right)\end{equation}
and

\begin{equation}\label{changedpert} \partial_t W w_V = -(1+V)^{1/2}\partial_x \left( (1+V)^{1/2}w_V + (1+V)w_V v_V + (1+V)\frac{w_V^2}{2}\right) \; .\end{equation}

Those equations are well posed in $H^1$ as long as $v$ and $v_V$ exist and have a priori bounds. Indeed, calling

$$f(t) = ||w(t)||_{H^\sigma} ||w(t)||_{H^1} \mbox{ and } f_V(t) = ||w_v(t)||_{H^\sigma} ||w_v(t)||_{H^1}$$
it comes that for $t\in [-T,T]$,

$$f(t) \leq ||w(t)||_{H^1}^2 \leq |\int_{0}^t \int_x v (\partial_x w) w | \leq ||v||_{X^0_T} \int_{0}^t ||w||_{H^1} ||w||_{L^\infty}$$
and thanks to Sobolev embedding theorem ($\sigma >1/2$)

$$f(t) \leq ||v||_{X^0_T} \int_{0}^t f(t)$$
also, with $C$ a constant independent from $V$,

$$f_V(t) \leq 2 \int_{0}^t \int_x w_V W w_V \leq C ||v_V||_{X^0_T} \int_{0}^t f_V(t) \; .$$

Since $w_0$ is in $H^1$, the equations \reff{changed} and \reff{changedpert} are well posed on $[-T,T]$ with initial datum $w_0$ and 

$$f(t),f_V(t) \leq e^{C\frac{|T|}{1+|T|}} ||w_0||_{H^1}||w_0||_{H^\sigma}\; .$$

Now, it appears that

$$||w_0||_{H^1}\leq N(u_0,T)^{1-s}||u_0||_{H^s} \mbox{ and } ||w_0||_{H^\sigma} \leq N(u_0,T)^{\sigma-s} ||u_0||_{H^s}^2$$
so

$$f(t),f_V(t) \leq C N(u_0,T)^{1+\sigma - 2s} ||u_0||_{H^s}\; .$$

The functions $u=v+w$ and $u_V = v_V+w_V$ are solution respectively of the BBM and the perturbed flow with initial datum $u_0$ and

$$||u(T)||_{L^2} \leq ||v||_{X^0_T}+ ||w(T)||_{L^2} \leq \frac{C}{1+|T|} + f(t)^{1/2} \leq C +C N^{(1+\sigma-2s)/2}||u_0||_{H^s} $$
so

$$||\psi(t)u_0||_{L^2}, \; ||\psi_V(t) u_0 ||_{L^2} \leq C + C N^{(1+\sigma-2s)/2}||u_0||_{H^s} \;.$$ \end{proof}

\begin{proposition}\label{clo}Let $s \in [0,\frac{1}{2}[$ and $\sigma \in ]\frac{1}{2}, 1]$. There exist $C$ such that for all $u_0 \in H^s$ and $T \in \R$,

$$||\psi_V(T) u_0 - \psi(T)u_0||_{L^2} \leq C ||V||_\infty (1 +  N^{1+\sigma-2s}||u_0||_{H^s}^2+ ||u_0||_{L^2}) e^{c(1 +  N^{(1+\sigma-2s)/2}||u_0||_{H^s}) |T|}$$
where $N = N(u_0, T)$ has been defined at the beginning of the subsection, by definition \ref{bign}.
\end{proposition}

\begin{proof} First, compute $||(K-D)(uv)||_{L^2}$ for all $u$, $v\in L^2$.

$$K-D = (1-(1+V)^{-1/2} ) K + (1+V)^{1/2}K (1 - (1+V)^{1/2})\; .$$

As $||(1+V)^{-1/2}-1||_{L^\infty} \leq C ||V||_\infty $ and $||(1+V)^{1/2}-1||_{L^\infty} \leq C ||V||_{\infty}$, the following result is ensured : for all $u,v\in L^2$,

$$||(K-D)(uv) ||_{L^2} \leq C ||V||_\infty ||u||_{L^2}\; ||v||_{L^2}\; .$$

Then, let $v \in L^2$ and $g(t) = ||e^{-tD} v - e^{-tK} v||_{L^2}$, then

$$g'(t) = ||(K-D)e^{-tK}v + D(e^{-tK}v-e^{-tD}v)||_{L^2} \leq ||K-D||_0 ||v||_{L^2} + ||D||_0 g(t)$$

$$g(t) \leq g(0) + ||K-D||_0 ||v||_{L^2} |t| e^{|t|\; ||D||_0}$$
and $g(0) = 0$. So, the following inequality applies : there exist $C,c$ two constants such for all $v\in L^2$, all $V \in \mathcal C^1$, and all $t\in \R$,

$$||e^{-tD}v - e^{-tK}v||_{L^2} \leq C||V||_\infty||v||_{L^2}  e^{c|t|}\; .$$

Write 

$$u_V = \psi_V(t) u_0 = e^{-t D }u_0 + \int_{O}^t e^{(s-t)D}D \frac{(1+V)^{1/2}u_V^2(s)}{2}ds$$
and

$$u = \psi(t) u_0 = e^{-tK}u_0 + \int_{0}^t e^{(s-t)K} K \frac{(\psi(s)u_0)^2(s)}{2}ds \; .$$

Let now 

$$f(t) = ||u_V - u||_{L^2}$$

$$f(t) \leq ||e^{-tD}u_0 - e^{-tK}u_0||_{L^2} + ||\int_{0}^t \left(  e^{(s-t)D}D \frac{(1+V)^{1/2}u_V^2(s)}{2} - e^{(s-t)K} K \frac{u^2(s)}{2}d\right) ||_{L^2}$$

The integral term is less than: 

\begin{tabular}{llll}
 & $||\int_{0}^t \left(  e^{(s-t)D}D \frac{(1+V)^{1/2}u_V^2(s)}{2} - e^{(s-t)K} D \frac{(1+V)^{1/2}u_V^2(s)}{2}\right) ||_{L^2} $ & $(=I)$  \\
$+$ & $||\int_{0}^t \left(  e^{(s-t)K}D \frac{(1+V)^{1/2}u_V^2(s)}{2} - e^{(s-t)K} K \frac{(1+V)^{1/2}u_V^2(s)}{2}\right) ||_{L^2} $ & $(=II)$  \\
$+$ & $||\int_{0}^t \left(  e^{(s-t)K}K \frac{(1+V)^{1/2}u_V^2(s)}{2} - e^{(s-t)K} K \frac{u_V^2(s)}{2}\right) ||_{L^2} $ & $(= III)$  \\
$+$ & $||\int_{0}^t \left(  e^{(s-t)K}K \frac{u_V^2(s)}{2} - e^{(s-t)K} K \frac{u^2(s)}{2}\right) ||_{L^2} $ & $(= IV)$  
\end{tabular}

Estimate now the different terms. For all $t\in [-T,T]$,

$$ I \leq \int_{0}^t C ||V||_\infty ||D(\frac{(1+V)^{1/2}(u_V)^2(s)}{2} e^{c |t-s|} ||_{L^2}ds\leq C ||V||_\infty (1 +  N^{1+\sigma-2s}||u_0||_{H^s}^2)e^{c|t|}$$

As $K$ is antisymmetric, $e^{tK}$ is isometric in $L^2$ so :

$$II \leq \int_{0}^t ||(D-K) \left( \frac{\sqrt{1+V}u_V^2(s)}{2} \right) ||_{L^2}ds \leq C ||V||_\infty |t| (1 +  N^{1+\sigma-2s}||u_0||_{H^s}^2)$$

$$II \leq C ||V||_\infty (1 +  N^{1+\sigma-2s}||u_0||_{H^s}^2)e^{c|t|}$$

$$III\leq \int_{0}^t \left(  ||K \frac{((1+V)^{1/2}-1)u_V^2(s)}{2} ||_{L^2}\right)\leq C ||V||_\infty |t| (1 +  N^{1+\sigma-2s}||u_0||_{H^s}^2)$$

$$III \leq C ||V||_\infty (1 +  N^{1+\sigma-2s}||u_0||_{H^s}^2)e^{c|t|}$$

$$IV \leq \int_{0}^t C ||u(s)+u_V(s)||_{L^2}||u(s)-u_V(s)||_{L^2}\leq C (1 +  N^{(1+\sigma-2s)/2}||u_0||_{H^s})\int_{0}^tf(s) ds \; .$$

To sum up, for all $t\in [-T,T]$, $f(t)$ is less than : 

$$f(t) \leq C||V||_\infty \left( ||u_0||_{L^2} + 1 +  N^{1+\sigma-2s}||u_0||_{H^s}^2\right)e^{c|t|} + C (1 +  N^{(1+\sigma-2s)/2}||u_0||_{H^s})\int_{0}^tf(s) ds$$

so

$$f(t) \leq C||V||_\infty \left( ||u_0||_{L^2} + 1 +  N^{1+\sigma-2s}||u_0||_{H^s}^2\right)e^{c|t|} e^{c(1 +  N^{(1+\sigma-2s)/2}||u_0||_{H^s}) |t|}$$

$$f(t) \leq C||V||_\infty \left( ||u_0||_{L^2} + 1 +  N^{1+\sigma-2s}||u_0||_{H^s}^2\right)e^{c(1 +  N^{(1+\sigma-2s)/2}||u_0||_{H^s}) |t|}\; .$$\end{proof}

\subsection{Evolution of the perturbed statistics}

Now see that the law of $\psi(t) u_0$ is not too different from the law of $u_0$. For this use the generating functional

$$Z_V(\lambda, t) = E_V(e^{i\langle \lambda , \psi(t)u_0}\rangle)$$
and prove that it is quite close to its initial data.

\begin{theorem} Let $\epsilon \in ]0,\frac{1}{2}[$. There exists $C,c$ such that for all $V \in \mathcal C^1$ with $||V||_\infty < \frac{1}{2}$, all $\lambda \in L^2$, and all $t\in \R$,

$$|Z_V(\lambda , t)-Z_V(\lambda , 0)| \leq C ||V||_\infty ||\lambda ||_{L^2} e^{c|t|^{6/\epsilon -3} }\; .$$
\end{theorem}

\begin{remark} The function $x\mapsto e^{ix}$ used in the generating functional could have been replaced by any $\mathcal C^1$ function $f$ whose derivative is bounded.\end{remark}

\begin{proof}First using the invariance of the measure $\mu_V$ through the flow $\psi_V$, it comes that for all $t\in \R$,

$$Z_V(\lambda, 0) = E_V(e^{i\langle \lambda , u_0\rangle})=  E_V(e^{i\langle \lambda , \psi_V(t)u_0\rangle})$$
so

$$|Z_V(\lambda , t)-Z_V(\lambda , 0)| = |E_V\left(e^{i\langle \lambda , \psi(t)u_0\rangle}- e^{i\langle \lambda , \psi(t)u_0\rangle} \right)|\; .$$

It appears then that

$$|Z_V(\lambda , t)-Z_V(\lambda , 0)|\leq E_V \left(|\langle \lambda , \psi(t) u_0 - \psi_V(t) u_0 \rangle | \right) \leq ||\lambda||_{L^2} E_V(||\psi(t) u_0 - \psi_V(t) u_0||_{L^2} ) \; .$$

As $\epsilon < \frac{1}{2}$, one can take $s$ in $]\frac{1+2\epsilon }{4}, \frac{1}{2}[$ and $\sigma = 2s-\epsilon > \frac{1}{2}$ such that $1+\sigma -2s = 1-\epsilon$. Apply then the proposition \reff{clo} with such $\sigma$ and $s$ :

$$|Z_V(\lambda , t)-Z_V(\lambda , 0)|\leq C ||\lambda||_{L^2} ||V||_\infty E_V \left( (1 +  N^{1-\epsilon}||u_0||_{H^s}^2+ ||u_0||_{L^2}) e^{c(1 +  N^{(1-\epsilon)/2}||u_0||_{H^s}) |t|}\right) $$

$$\leq C ||\lambda||_{L^2} ||V||_\infty e^{c |t|}E_V \left( (1 +  N^{1-\epsilon}||u_0||_{H^s}^2+ ||u_0||_{L^2}) e^{c N^{(1-\epsilon)/2}||u_0||_{H^s} |t|}\right) \; .$$

Remember that $N = N(u_0,|t|)$ is the smallest integer such that the solutions of BBM and of the perturbed BBM are given by local well posedness on $[-|t|,|t|]$ with initial datum $(1-\Pi_N)u_0$. 

Using that $N^{(1-\epsilon)/2}||u_0||_{H^s} \leq N^{1-\epsilon /2} + ||u_0||_{H^s}^{2-\epsilon}$ the mean value : 

$$E_V \left( (1 +  N^{1+\sigma-2s}||u_0||_{H^s}^2+ ||u_0||_{L^2}) e^{c N^{(1-\epsilon)/2}||u_0||_{H^s} |t|}\right)\leq \sqrt{I.1\; I.2} \sqrt{II.1\; II.2}\sqrt{III.1\; III.2}$$
with

$$I.1 = E_V (e^{2c ||u_0 ||_{H^s}^{2-\epsilon} |t|}) , II.1 = E_V (||u_0||_{H^s}^4 e^{2c ||u_0 ||_{H^s}^{2-\epsilon} |t|}) , III.1= E_V (||u_0||_{L^2}^2e^{2c ||u_0 ||_{H^s}^{2-\epsilon} |t|})$$
and

$$I.2 = III.2 = E_V(e^{2c N^{1-\epsilon / 2}|t|}) , II.2= E_V(N^{2-2\epsilon}e^{2c N^{1-\epsilon / 2}|t|})\; .$$

First, remember that there exists $c'> 0$ and $C'$ such that

$$\mu_V \left( \lbrace u_0 \; |\; ||u_0 ||_{H^s} > R \rbrace\right) \leq C' e^{-cR^2}$$
so

$$I.1,II.1, III.1 \leq C_\epsilon e^{c_\epsilon |t|^{2/\epsilon}}\; .$$

Then, 

$$P(N> N_0 ) = \mu_V \left( \lbrace u_0 \; | \; ||(1-\Pi_{N_0})u_0 ||_{L^2} \geq \frac{1}{C(1+|t|)} \rbrace \right) \leq C'e^{-c' N_0/(1+|t|)^2}$$
so

$$I.2,II.2,III.2 \leq C_\epsilon e^{c_\epsilon |t|^{6/\epsilon -3}}\; .$$

Since $\epsilon < \frac{1}{2}$, it appears that $\frac{6}{\epsilon} -3 > \frac{2}{\epsilon}$, so in the end , ther exist two constants $C_\epsilon, c_\epsilon$ such that

$$|Z_V(\lambda,t) - Z_V(\lambda, 0)|\leq C_\epsilon ||\lambda||_{L^2} ||V||_\infty e^{c_\epsilon |t|^{\frac{6}{\epsilon}-3}}\; .$$\end{proof}

\begin{remark}The averages of the products of the amplitudes admit the same kind of estimates. For instance, calling 

$$(\alpha_V^2)_{n,m}(t) = E_V(\langle c_n,\psi(t)u_0\rangle \; \langle c_m,\psi(t)u_0 \rangle) \; ,$$
it appears that

$$|(\alpha_V^2)_{n,m}(t) - (\alpha_V^2)_{n,m}(0)|\leq C'_\epsilon||V||_\infty e^{c'_\epsilon |t|^{\frac{6}{\epsilon}-3}}$$
only with different constants.
\end{remark}

\section*{Acknowledgements}

This work is supported by part by the ERC project Dispeq. It is also a part of a PhD thesis under the supervision of Nikolay Tzvetkov. 

We would like to thank Baptiste Billaud for his valuable help and insight regarding the physics of wave turbulence.

\bibliographystyle{amsplain}
\bibliography{abibbbm} 
\nocite{*}

\end{document}